\PassOptionsToPackage{unicode}{hyperref}
\PassOptionsToPackage{hyphens}{url}

\documentclass[
  12pt,
]{article}
\usepackage{lmodern}
\usepackage{amssymb,amsmath}
\usepackage{ifxetex,ifluatex}
\ifnum 0\ifxetex 1\fi\ifluatex 1\fi=0 % if pdftex
  \usepackage[T1]{fontenc}
  \usepackage[utf8]{inputenc}
  \usepackage{textcomp} % provide euro and other symbols
\else % if luatex or xetex
  \usepackage{unicode-math}
  \defaultfontfeatures{Scale=MatchLowercase}
  \defaultfontfeatures[\rmfamily]{Ligatures=TeX,Scale=1}
\fi
\IfFileExists{upquote.sty}{\usepackage{upquote}}{}
\IfFileExists{microtype.sty}{% use microtype if available
  \usepackage[]{microtype}
  \UseMicrotypeSet[protrusion]{basicmath} % disable protrusion for tt fonts
}{}
\makeatletter
\@ifundefined{KOMAClassName}{% if non-KOMA class
  \IfFileExists{parskip.sty}{%
    \usepackage{parskip}
  }{% else
    \setlength{\parindent}{0pt}
    \setlength{\parskip}{6pt plus 2pt minus 1pt}}
}{% if KOMA class
  \KOMAoptions{parskip=half}}
\makeatother
\usepackage{xcolor}
\IfFileExists{xurl.sty}{\usepackage{xurl}}{} % add URL line breaks if available
\IfFileExists{bookmark.sty}{\usepackage{bookmark}}{\usepackage{hyperref}}
\hypersetup{
  pdftitle={Projection Bias in Effort Choices},
  hidelinks,
  pdfcreator={LaTeX via pandoc}}
\title{Projection Bias in Effort Choices\thanks{Central European University, Department of Economics and Business.
Email: kaufmannm@ceu.edu. I am grateful to Matthew Rabin, David Laibson,
Gautam Rao, and Josh Schwartzstein. I also thank Ben Bushong, Krishna
Dasaratha, Anastassia Fedyk, Tristan Gagnon-Bartsch, James Hodson,
Botond K\H{o}szegi, Annie Liang, Ben Lockwood, Neil Thakral, Linh T. Tô,
and Gal Wettstein, as well as seminar participants at CEU, Harvard,
Humboldt, the Hungarian Academy of Sciences, Liser, and LMU for helpful
comments.}}

\author{
Marc Kaufmann (Central European University)
}

\date{\today}

\usepackage{textcomp,empheq}
\usepackage{bbm}

\usepackage{amsthm}
\newtheorem{proposition}{Proposition}
\newtheorem{theorem}{Theorem}

\newtheorem{definition}{Definition}

\newtheorem{lemma}{Lemma}
\newtheorem{corollary}{Corollary}

\usepackage{mathtools}

\newcommand{\vect}[1]{\boldsymbol{#1}}

\providecommand{\JEL}[1]
{
  \small
  \textbf{\textit{JEL ---}} #1
}

\usepackage{setspace}
\singlespace
\usepackage[margin=3.2cm]{geometry}

\begin{document}
\maketitle
\begin{abstract}
Working becomes harder as we grow tired or bored. I model individuals
who underestimate these changes in marginal disutility -- as implied by
``projection bias'' -- when deciding whether or not to continue working.
This bias causes people's plans to change: early in the day when they
are rested, they plan to work more than late in the day when they are
rested. Despite initially overestimating how much they will work, people
facing a single task with decreasing returns to effort work optimally.
However, when facing multiple tasks, they misprioritize urgent but
unimportant over important but non-urgent tasks. And when they face a
single task with all-or-nothing rewards (such as being promoted) they
start, and repeatedly work on, some overly ambitious tasks that they
later abandon. Each day they stop working once they have grown tired,
which can lead to large daily welfare losses. Finally, when they have
either increasing or decreasing productivity, people work less each day
than previously planned. This moves people closer to optimal effort for
decreasing, and further away from optimal effort for increasing
productivity.
\end{abstract}

\JEL{D03, J22}

\hypertarget{introduction}{%
\section{Introduction}\label{introduction}}

Our tastes fluctuate, often rapidly: we grow tired and thirsty from
running, and we savor food or crave coffee more the longer we go
without. Furthermore, evidence from a variety of domains suggests that
our perceptions of our tastes are biased towards our current tastes: we
misperceive future tastes as being closer to our current tastes than
they will be (Loewenstein, O'Donoghue, and Rabin (2003)).\footnote{Evidence
  for projection bias has been found for food (Read and Van Leeuwen
  (1998); Nordgren, Pligt, and Harreveld (2008)), drink (Van Boven and
  Loewenstein (2003)), sexual arousal (Loewenstein, Nagin, and
  Paternoster (1997); Ariely and Loewenstein (2006)), effortful tasks
  (Augenblick and Rabin (2019)), heroin substitute cravings (Badger et
  al. (2007)), the endowment effect (Loewenstein and Adler (1995)) and
  for predictions of gym attendance (Acland and Levy (2015)). Projection
  bias resembles immune neglect (Gilbert et al. (1998)) whereby people
  overestimate how long they will feel bad about negative events.
  Bushong and Gagnon-Bartsch (2020) additionally find evidence for
  \emph{interpersonal} projection bias for choices over effort.} This
\emph{projection bias} can trigger undesirable and unintended habits and
behaviors, such as buying too much when shopping on an empty stomach, or
becoming addicted due to under-appreciating the future intensity of
cravings. In this paper, I study effort choices where the distaste for
work changes, such as when students grow bored of studying or employees
become tired of working. Due to projection bias, individuals mispredict
future disutility of work, which can cause them to mis-prioritize
between tasks, waste time on never-to-be-completed tasks, and
inefficiently choose when to work on those tasks.

I provide and discuss the formal model in Section 2. The agent works
continuously on a single task in one period and stops working the moment
she \emph{perceives} stopping to be optimal given the monetary benefits
and the perceived disutility from work. Letting \(s\) denote the total
time she has worked so far that day, the instantaneous disutility from
continuing to work is equal to \(D'(S)\). She misperceives her future
disutility due to projection bias: she predicts that her marginal
disutility after \(e\) hours of work lies between her current marginal
disutility, \(D'(s)\), and her true future marginal disutility,
\(D'(e)\). So when she is rested and marginal disutility is thus
particularly low, she overestimates how easy it will be to work, while
when tired and marginal disutility is thus particularly high, she
underestimates it. This leads to changing and inconsistent plans which
in turn can lead to the time-inconsistent behavior that I study in this
paper.

In Section 3, I analyze the implications from underestimating the
disutility of future work in single-period, single-task settings.
Because I assume that a person grows tired the longer she works, she
underestimates how unpleasant the work will be later that day and
overestimates how much she will work. Despite this, with decreasing
returns to effort, she works optimally. She works until the marginal
benefits equal her marginal disutility, and then stops. I then consider
all-or-nothing rewards, for which she receives known, fixed benefits if
she completes the task by the end of a given day. A projection-biased
person starts some overly costly all-or-nothing tasks, but as work
becomes more unpleasant, she realizes some of the higher costs and may
give up. If she does give up, she is better off than if she had
committed herself to finishing the task, because she still overestimates
how much she should work.

Even for decreasing returns to effort she only works optimally if she
faces a \emph{single} task. If, as I assume in Section 4, she has to
allocate her time across two consecutive tasks, she spends too much time
on the first task compared to the second: when she switches from the
first to the second task, she overestimates how much she will work on
the second task. If she realized how much she will actually work, she
would switch earlier. This suggests more generally that, when
multi-tasking, projection bias leads people to work too much on early
stages of tasks or on urgent but unimportant tasks, compared to later
stages or flexible but important tasks.

In Section 5, I extend the basic model to consider multi-day settings
with all-or-nothing tasks that are due on some fixed future day. I
simplify the analysis by considering a continuous-time setup with a
continuum of days. At the start of each day, the person may plan to
complete the task efficiently, yet stop working earlier than anticipated
once she has grown tired, in which case she plans to drop the task
entirely. This repeated starting and stopping -- due to naiveté about
her bias -- leads to one of two outcomes. Either she completes the task
despite working too little in the early days, but has to work harder
later on. Or she repeatedly wastes time on working on a task that she
eventually drops for good. In fact, as long as benefits are insufficient
to lead the person to complete the task, increasing benefits makes her
worse off, because they lead her to waste even more time across more
days for no change in the outcome. Even an almost unbiased person can
thus be induced to complete almost all the work, yet fail to complete
the task.

In Section 6, I consider a setting in which the main benefit from
working consists in having less work in the future, so that the
perceived value of the benefit fluctuates with current tiredness.
Specifically, the person has to produce a fixed total \emph{output} by a
given day, but her productivity either increases or decreases over time.
Then, the benefit of more work today is to have less work in the future.
On low-productivity days, the person works little, and therefore
perceives the benefits from saving, say, \(5\) hours of future work as
low. On high-productivity days, the person works a lot and therefore
perceives the same \(5\) hours of future work more costly and hence more
valuable to save them. Thus the person works too much on
high-productivity days and too little on low-productivity days.
Moreover, when productivity changes monotonically, whether it is
increasing or decreasing, people work less on a given day than they
planned at the end of the previous day. This change in plans moves them
closer to optimal behavior when productivity is decreasing, but further
from optimal behavior when productivity is increasing, which highlights
again that committed choices are not always preferable to on-the-spot
decisions.

My paper is most closely related to Loewenstein, O'Donoghue, and Rabin
(2003) who formalize the model of projection bias and apply it to
durable goods consumption, the endowment effect, and habit formation. I
focus on the time-inconsistent choices resulting from changes in fatigue
in effort choices, an economically important domain. Since agents
repeatedly reoptimize as they grow tired or rested, the final success of
work depends on combining the many decisions made while having
inconsistent plans.\footnote{Herrnstein and Prelec (1991)'s model of
  melioration under distributed choices is somewhat related, but differs
  substantially in the sense that plans don't matter in their model,
  while they are central to the results in my paper.} The constant
reoptimization in the multi-day model is closest to the model of
instantaneous gratification of Harris and Laibson (2013), including in
using ordinary differential equation in continuous time to approximate
the discrete-time optimization of naive agents.\footnote{Like Harris and
  Laibson (2013), I assume naiveté: a person is unaware of her bias,
  which allows for repeated changes in plans. It is beyond this paper to
  study how a person could be aware of her changing plans without being
  aware of the underlying reason for the changes.} This approach may
help integrate projection bias into specific economic settings, in
particular to time (mis)management and personnel economics. For
instance, Buehler, Griffin, and Ross (1994) find that students believe
that they will finish their bachelor's thesis earlier than they actually
do -- which they explain by students underestimating the number of hours
necessary for the task. Projection bias provides a complementary
explanation that people overestimate how much they will work. Regarding
my results on multi-tasking, Coviello, Ichino, and Persico (2015) and
Bray et al. (2016) find empirical evidence that it decreases
productivity. Coviello, Ichino, and Persico (2014) show that workers may
engage knowingly in intrinsically inefficient multi-tasking due to
lobbying by co-workers and superiors. With projection bias, even if
multi-tasking is not intrinsically inefficient, such lobbying will lead
workers to multi-task inefficiently. Given the possibly large welfare
losses under inconsistent behavior, my results highlight the potential
to expand the study of projection bias beyond domains with large swings
in taste (Levy (2009); Chaloupka, Levy, and White (2019)) and binding
choices (Conlin, O'Donoghue, and Vogelsang (2007); Busse et al. (2015);
Buchheim and Kolaska (2017); Michel and Stenzel (2020)).

\hypertarget{a-model-of-projection-bias-in-effort-choices}{%
\section{\texorpdfstring{A Model of Projection Bias in Effort Choices
\label{sec:setup}}{A Model of Projection Bias in Effort Choices }}\label{a-model-of-projection-bias-in-effort-choices}}

\subsection{The Formal Model}

\paragraph{Environment}

Consider a baseline setup with a single period during which a person
works for a single block of \(e\) hours. She earns a monetary benefit
\(B(e)\) and incurs a disutility \(D(e)\) for this work, where
\(D(\cdot)\) is continuously differentiable, with \(D(0) = 0\),
\(D'(e) \geq 0\), and \(D''(0) \geq 0\). The marginal disutility
\(D'(e)\) is the instantaneous disutility of continuing to work at
\emph{time \(e\)} -- at a time when the person has worked for a duration
\(e\).

\paragraph{Perceived Disutility}

Projection bias as defined by Loewenstein, O'Donoghue, and Rabin (2003)
leads a person to misperceive her future taste for work as more similar
to her current taste for work than it will actually be. As people work,
they grow more tired of working which is captured by increasing marginal
disutility, and so they perceive future work as more onerous the more
tired they currently are. Formally:

\begin{definition}[Projection Bias]
At time $s$ -- when a person has worked for a time $s$ -- the person mispredicts marginal disutility at a future time $e$ to be 
\begin{equation}
    \tilde{D}'(e | s) = (1 - \alpha) D'(e) + \alpha D'(s)
    \label{eq:define-projection-bias}
\end{equation}
where $\alpha \in [0, 1]$ is the degree of projection bias. Moreover, she is \emph{naive} with respect to her projection bias: she does not realize that her perception depends on $s$ and thus believes that she always perceives the marginal disutility as she does currently.
\end{definition}

With this definition the perceived disutility from total effort \(e\) is
\(\tilde{D}(e|s) = ( 1 - \alpha ) \cdot D(e) + \alpha \cdot D'(s) \cdot e\),
since \(D(0) = 0\). Note that, since the taste for money (for consuming
goods that money can buy) does not change substantially nor
systematically while working, the benefits \(B(\cdot)\) are being
perceived correctly.

\paragraph{Behavior}

I assume that the person stops working once she perceives working as
suboptimal, which then determines the total amount of work done. Since
the actual and perceived utility from stopping at time \(s\) is
\(B(s)\), and the perceived utility of continuing until time \(e > s\)
is \(B(e) - (\tilde{D}(e|s) - \tilde{D}(s|s))\), we have the following
behavior, which I describe as \emph{momentary} work decisions:

\begin{definition}[Momentary Work Decision]\label{def:single-task-decision}
  A projection-biased person works until $\tilde{e}^{*}$ given by 
  \begin{equation*}
    \tilde{e}^{*} = \inf \{s: B(e) - \tilde{D}(e|s) < B(s) - \tilde{D}(s|s) \text{,  } \forall e > s\}
  \end{equation*}
\end{definition}

\subsection{Discussion of the Model}

Projection bias leads to projecting marginal disutility under two
assumptions. First, at each instant the person makes a single binary
decision of whether or not to work. If the person could choose the level
of work intensity at each moment, then the instantaneous disutility
would depend on the level of intensity chosen at time \(s\), so that
\(D'(s)\) could no longer capture the instantaneous disutility. Second,
since projection bias is not about projecting indirect decision
utilities, but about projecting immediate hedonic utility, I assume that
working at time \(s\) incurs the \emph{instantaneous hedonic} disutility
equal to \(D'(s)\).

Now consider a projection-biased person who decides \emph{momentarily}
-- moment-by-moment -- whether or not to work right now, but who cannot
commit to how much she will work in total. If, as I assume, the person
is naive with respect to their own projection bias and hence does not
anticipate that she will deviate from her current plan,\footnote{The
  evidence from (Read and Van Leeuwen 1998) shows that, despite
  experience with fluctuations in hunger, adults still display
  projection bias over hunger. A nice experiment by Le Yaouanq and
  Schwardmann (2019) in the context of present bias shows however that
  participants do make less overoptimistic predictions about their own
  future behavior after making initial predictions and receiving
  feedback on it. This shows that we need richer models of learning
  about one's biases.} then she will only work as long as she perceives
it optimal.\footnote{In case the agent is indifferent between working
  and stopping, the model assumes that the agent works, since the person
  only stops if it is strictly suboptimal to continue working. In the
  settings that I study, this happens only for non-generic edge-cases.}
This determines total effort that period under the assumption that there
is no opportunity for resting within a period -- which I exclude by
assuming that all the work is being done in one block. If resting were
possible, the person might decide to resume work after a break. Of
course, both resting and intensity of work are important, since a
projection-biased person may misoptimize both, but I study them in
separate ongoing work.

Finally notice that when \(\alpha = 0\) the person has no projection
bias, and the actual work done \(\tilde{e}^{*}\) equals the optimal work
\(e^{*}\), if it is unique.\footnote{If \(s < e^{*}\), then
  \(B(e^{*}) - \tilde{D}(e^{*}|s) - (B(s) - \tilde{D}(s|s)) = B(e^{*}) - D(e^{*}) - (B(s) - D(s)) > 0\),
  since \(\tilde{D}(e|s) = D(e)\) and since \(e^{*}\) maximizes
  \(B(e) - D(e)\). Hence \(\tilde{e}^{*} \geq e^{*}\). But similarly,
  for \(s = e^{*}\), then
  \(B(s) - D(s) = B(e^{*}) - D(e^{*}) > B(e) - D(e)\) for every
  \(e > e^{*} = s\), hence \(\tilde{e^{*}} \leq s = e^{*}\). Thus
  \(\tilde{e}^{*} = e^{*}\).} Thus the setup nests the unbiased case.

\paragraph{Related Literature}

In terms of modeling projection bias, Loewenstein, O'Donoghue, and Rabin
(2003) define and formalize projection bias as the general tendency for
people to perceive their future tastes to be more similar to their
current tastes than they are, when tastes in any given moment depend on
some state \(s\). They focus on habit formation, durable goods, and loss
aversion. Gagnon-Bartsch and Bushong (2019) develop a model in which
people mislearn about their preferences for the good as they learn about
it from experience. People enjoy a good or activity more when it turns
out better than expected, and less when it turns out worse than
expected, but misattribute these feelings of elation or disappointment
to the good itself. Mislearning under projection bias instead may lead
people to neglect how much their state affected the enjoyment of a good,
thus misattributing their current state partially to the consumed good.
Haggag et al. (2019) develop a simple model of such state misattribution
and find evidence for it in two consumer decisions.

My simplest setup where the person decides when to stop working on a
single task is similar to stopping problems under time-inconsistent
preferences, such as Quah and Strulovici (2013), Hsiaw (2013), and Huang
and Nguyen-Huu (2018). In these models, the time inconsistency stems
from present bias or dynamically inconsistent changes in patience,
rather than from state changes that depend themselves on earlier
decisions. In the multi-tasking and multi-period settings that I study,
the agent no longer faces a single stopping decision, but several
decisions that jointly determine the final outcome. This setup is thus
more closely related to Harris and Laibson (2013) and Ahn, Iijima, and
Sarver (2020), both of which allow for repeated decisions to combine
over time to a final outcome such as final savings.

Having defined projection bias and the person's decision problem, the
next sections explore how the inconsistent plans of a projection-biased
person due to changing tastes affect their work decisions.

\hypertarget{single-task-choices}{%
\section{Single-Task Choices}\label{single-task-choices}}

Let us start with single-day decisions where people maximize their
utility in momentary work decisions over a single task as described in
section 2: people who have worked for \(s\) hours so far keep working if
they perceive it optimal \emph{at time \(s\)}. When the disutility
\(D(\cdot)\) is convex and the benefits \(B(\cdot)\) are linear or
concave, a projection-biased person works optimally -- despite (in fact,
because of) her changing plans. Nonetheless, such a person has
overoptimistic beliefs about how much she will work. I then consider
all-or-nothing tasks, where a person receives a known reward if she
completes a minimum amount of work. I show that people start overly
ambitious tasks, so that they either end up completing the task despite
it not having been worthwhile, or they quit the task without receiving
any benefit for their effort.

\hypertarget{optimal-behavior-and-optimistic-beliefs-with-convex-disutility-and-linear-benefits}{%
\subsection{Optimal Behavior and Optimistic beliefs with Convex
Disutility and Linear
Benefits}\label{optimal-behavior-and-optimistic-beliefs-with-convex-disutility-and-linear-benefits}}

Consider Anna, a projection-biased student with \(\alpha = 0.5\), who
has an exam tomorrow. The benefits of every additional hour of studying
are equal to \(3\), and studying becomes more unpleasant the longer she
studies. Specifically, Anna's daily disutility is quadratic in total
time studied, thus \(D(e) = \frac{e^2}{2}\) and \(D'(e) = e\). After
having studied for \(s\) hours, Anna plans to study until her currently
perceived marginal disutility is equal to her marginal benefits (which
are constant and equal to \(3\)). I denote the time at which she plans
to stop by \(\tilde{e}^{*}(s)\), the total hours she plans to work after
having worked for \(s\) hours. She perceives her marginal disutility
after studying for \(e\) hours to lie between her current marginal
disutility, \(D'(s)\), and her actual marginal disutility after \(e\)
hours of studying, \(D'(e)\):

\begin{equation*}
    \underbrace{\tilde{D}'(e | s)}_{\mathclap{\text{Perceived }D'}} = (1 - \alpha) \overbrace{D'(e)}^{\mathclap{\text{Actual }D'}} + \alpha \underbrace{D'(s)}_{\mathclap{\text{Current }D'}} = \frac{1}{2} (D'(e) + D'(s))
\end{equation*}

At the start of the day, Anna hasn't studied at all and \(s = 0\). So
she thinks that her marginal disutility after \(e\) hours of studying
will be \(\tilde{D}'(e|0) = \frac{1}{2} D'(e)\). She plans to work for
\(\tilde{e}^{*}(0)\) hours, with
\(\tilde{D}'(\tilde{e}^{*}(0)|0) = 3 \implies 1/2 \cdot (D'(\tilde{e}^{*}(0)) + 0) = \frac{1}{2}\tilde{e}^{*}(0) = 3 \implies \tilde{e}^{*}(0) = 6\).
Anna plans to study for \(6\) hours and thus starts studying. After
\(2\) hours of studying, the current marginal disutility is
\(D'(2) = 2\). Anna now plans to study for \(\tilde{e}^{*}(2)\) hours in
total, with \(\tilde{D}'(\tilde{e}^{*}(2)|2) = 3\) -- the first order
condition as she perceives it now. This leads to
\(1/2 \cdot (D'(\tilde{e}^{*}(2)) + D'(2)) = 3 \implies \tilde{e}^{*}(2) = 4\)
hours. Finally, once she has completed \(3\) hours of studying, the
current marginal disutility is \(D'(3) = 3\), so that
\(\tilde{e}^{*} = 3\) and Anna stops studying.

The same logic applies when the returns to effort are decreasing rather
than constant, which leads to Proposition \ref{prop:optimal-effort}. All
proofs can be found in the appendix.

\begin{proposition}
\label{prop:optimal-effort}
Let $D(.)$ be a strictly convex function with $D'(\cdot) \to \infty$, let $\alpha \in [0, 1)$, and $B(.)$ be both differentiable and linear or concave. Then a projection-biased person who makes momentary work decisions works optimally. Moreover, letting $e^{*}$ be the optimal amount of work and $\tilde{e}^{*}(s)$ be the perceived optimal amount of work after $s$ hours of work, then $\tilde{e}^{*}(s) > e^{*}$ $\forall s < e^{*}$. 
\end{proposition}

Proposition \ref{prop:optimal-effort} relies on momentary work
decisions. If Anna had to make an irreversible (or hard-to-reverse)
choice, then she would choose to work too much. This is not likely in
the case of studying, but may be the case if Anna is grading exams for a
course or working on a common project with a friend. In such situations,
due to being overoptimistic, Anna will overcommit to working too much.

Proposition \ref{prop:optimal-effort} also highlights that Anna
constantly overestimates how much she will work. Why? By assumption, the
marginal disutility of effort increases, so that Anna -- who projects
her current marginal disutility -- underestimates how high marginal
disutility will be later that day when she stops, and therefore
overestimates how long she will study.

This strongly limits the scope of the result: while behavior is optimal
when the disutility is convex and benefits are linear or concave, the
beliefs over future work are overoptimistic. As long as overoptimistic
beliefs don't affect other decisions, everything is well. However, as
soon as some decisions rely on predictions of future effort, mistakes
will be made. I highlight this major caveat in Proposition
\ref{prop:multi-tasking-single-day} in the next section. Now, I consider
tasks with all-or-nothing benefits that are received only if the person
completes a minimum number of hours.

\hypertarget{all-or-nothing-tasks}{%
\subsection{\texorpdfstring{All-or-Nothing
Tasks\label{subsec:fixed-hours}}{All-or-Nothing Tasks}}\label{all-or-nothing-tasks}}

\begin{definition}[All-or-Nothing Task]
A single-period \emph{all-or-nothing} task $(E_{0}, B_{0})$ is a task that pays benefits $B_0$ if the person works at least $E_0$ hours by the end of the period, and pays $0$ otherwise: $B(e) = B_{0} \cdot \mathbbm{1}(e \geq E_{0})$.
\end{definition}

Each instant, the person chooses whether to start or continue the task.
She does so if she \emph{currently} thinks that completing the task is
better than quitting the task. Suppose that Alice, a projection-biased
high-school student with \(\alpha = 0.5\), has a deadline to finish a
college application tonight, which will take her \(6\) hours. With
quadratic disutility \(D(e) = e^{2}/2\), we have \(D(6) = 18\), and let
us suppose that \(B = 12\), so that Alice should not do the application.
Thus an unbiased person does not do start nor complete the task.

Does Alice start the application and, if so, does she finish it? She
starts if the perceived disutility at \(s = 0\), \(\tilde{D}(6 | 0)\),
is less than \(B\), where
\(\tilde{D}(6 | 0) = (1 - \alpha) D(6) + \alpha D'(0) \cdot 6 = 9 + \alpha D'(0) \cdot 6\).
Since \(D(e) = \frac{e^2}{2}\), then \(\tilde{D}(6 | 0) = 9 < 12 = B\)
and Alice starts the application. Now imagine what happened if Alice
worked for another two hours -- which, as we will see, does not happen.
Then the perceived disutility of completing the application would be
\(\tilde{D}(6|2) - \tilde{D}(2|2) = 13 > 12 = B\), so she would have
stopped working before reaching two hours of work. The reason is that
Alice perceives the final \(4\) hours of work as so much more unpleasant
after \(2\) hours of working than at the start of the day that she
perceives the task no longer worth completing -- despite having less
work left to do. Proposition \ref{prop:convex-single-day} states
formally when this happens.

\begin{proposition}
\label{prop:convex-single-day}[All-or-Nothing]
Let $D(.)$ be a strictly convex function with $D'(\cdot) \to \infty$, let $\alpha \in [0, 1)$. Let $\tilde{e}^{*}(E_{0}, B_{0})$ be the actual effort exerted by a person making momentary work decisions for a single-period all-or-nothing task. Then there exists a unique $E_H \geq 0$ such that the following statements hold:
\begin{enumerate}
    \item $\forall E_0$, if $B \in (\tilde{D}(E_{0}|0), D(E))$ then $\tilde{e}^{*}(E_{0}, B_{0}) > 0$, yet $B_{0} - D(E_{0}) < 0$
    \item $\forall E_{0} < E_H$ if $\tilde{e}^{*}(E_{0}, B_{0}) > 0$ then $\tilde{e}^{*}(E_{0}, B_{0}) = E_0$.
    \item $\forall E_{0} > E_H$, $\exists B_{0}(E)$ s.t. $0 < \tilde{e}^{*}(E_0, B_0) < E_0$.
\end{enumerate}
\end{proposition}

The proposition states that, first, the person starts tasks that are not
worth doing; that if the task requires sufficiently low effort, every
task that is started is completed; and if the task requires sufficiently
high effort, there is a task-specific payment such that the person
starts the task yet fails to finish it, thus wasting effort for no
benefit. Note that the first two points imply that all worthwhile tasks
are started and finished. Moreover, when \(D'(0) = 0\), then
\(E_H = 0\), so that starting and stopping can happen for all tasks.
Proposition \ref{prop:convex-single-day} applies more generally to tasks
with sufficiently convex benefits, not just all-or-nothing tasks. Thus
it also applies to situations with few discrete outcomes, such as
promotions or grades on exams, where the probability of success is
S-shaped and hence convex for each outcome.

The result that people start but don't finish a project that they start
superficially resembles the result by O'Donoghue and Rabin (2008).
There, however, people procrastinate on a project with no deadline, and
therefore never finish it, expecting to do so eventually, whereas here
they start and stop because of the close deadline. Thus both the
predictions on planned behavior and the welfare implications of the
results differ: while commitment would benefit a naive procrastinator,
it would hurt a projector.

\hypertarget{multi-tasking-with-concave-benefits}{%
\section{Multi-Tasking with Concave
Benefits}\label{multi-tasking-with-concave-benefits}}

Let us revisit the situation with convex disutility and decreasing
returns to effort, but with a twist: the person now divides her time
between two tasks, each of which has decreasing returns to effort.

\subsection{Multi-Tasking Model}

\paragraph{Environment}

A projection-biased person works on \(T\) consecutive tasks in a single
period: she works on task 1 for a duration \(e_{1}\), then switches to
task 2 for a duration \(e_{2}\) and so on until working on task \(T\)
for a duration \(e_{T}\). The disutility depends on total effort, so
that we can write it as \(D(\sum_{t = 1}^{T} e_{t})\), while the
monetary benefits are task-specific, so that we can write them as
\(\sum_{t = 1}^{T} B_{t}(e_{t})\).

\paragraph{Behavior}

Consider a person who is working on the \(i\)th task and let \(E_{i-1}\)
be the total amount of effort exerted on the first \(i-1\) tasks, which
can no longer be changed. Let \(V_{i}(e_{i}, E_{i-1}|s)\) denote the
perceived value from the remaining tasks \(i\) through \(T\) when
planning to put total effort \(e_i\) on task \(i\), perceived when the
person's current tiredness is \(s\) -- when the person has exerted total
effort \(s\) during this period. Then this is given as follows:

\begin{definition}[Perceived Continuation Value]
    Let $E_{i-1}$ be the total effort exerted on past tasks $1$ through $i-1$. Then 
    \begin{equation*}
        V_{i}(e_i, E_{i-1}|s) = \max_{ (e_j)_{j > i} } \sum_{t = i}^T B(e_t) - (D(E_{i-1} + \sum_{t = i}^T e_t) - D(E_{i-1}))
    \end{equation*}
\end{definition}

The person stops working on task \(i\) and switches to task \(i + 1\)
once she perceives it as strictly suboptimal to continue working on task
\(i\): when \(V_{i}(e_{i}, E_{i-1}|s) < V_{i}(s - E_{i-1}, E_{i-1}|s)\)
for all \(e_{i} > s - E_{i-1}\), since if she stops right away, she will
have spent a time \(s - E_{i-1}\) on task \(i\). This leads to the
following generalization of Definition \ref{def:single-task-decision}:

\begin{definition}[Multiple Effort Decisions]\label{def:multi-task-decision}
  A projection-biased person who works on $T$ consecutive tasks in a single period exerts effort $\tilde{e}_{i}^{*}$ on task $i$ given by
  \begin{equation*}
    \tilde{e}_{i}^{*} = \inf \{s: \tilde{V}_i(e_i, E_{i-1}|s) < \tilde{V}_i(s - E_{i-1}, E_{i-1}|s) \text{,  } \forall e_i > s - E_{i-1}\}
  \end{equation*}
  where $E_{i-1} = \sum_{t = 1}^{i - 1} \tilde{e}_{t}^{*}$.
\end{definition}

\subsection{Multi-Tasking Results}

To illustrate, suppose that Elaine has two problem sets due the same
day, one in economics due at 3pm and one in mathematics due at 8pm.
Given these deadlines, she starts working on the economics problem set
first. For simplicity, assume that the benefits for each problem set are
the same and given by \(B(.)\), which has decreasing marginal returns.
After working on the first problem set for \(s\) hours, she plans to
spend \(\tilde{e}^{*}(s)\) hours on each assignment. She thus stops
working on the economics assignment when she thinks that she has done
half the work. Let's say that this happens after \(5\) hours, at which
point she expects to do another \(5\) hours on the mathematics
assignment. She is of course wrong, and overestimates how long she will
keep working. Thus she may stop working after only \(3\) hours on the
mathematics assignment. We know from Proposition
\ref{prop:optimal-effort} that this choice is optimal \emph{conditional}
on her having spent \(5\) hours on the economics assignment -- so the
mistake she makes is to spend too much time on the economics assignment,
because she overestimates at that time how much she will work on the
mathematics assignment.

\begin{proposition}
\label{prop:multi-tasking-single-day}
There are two tasks with strictly concave and continously differentiable benefits $B_1(.)$ and $B_2(.)$ that have to be completed one after the other in that order. Let $\tilde{e}_1^{*}$ and $\tilde{e}_2^{*}$ be the actual effort spent on the two tasks, and $e_1^{*}$ and $e_2^{*}$ be the optimal effort levels. Then $B'(\tilde{e}_2^{*}) > B'(e_2^{*}) = B'(e_1^{*}) > B'(\tilde{e}_1^{*})$ and $\tilde{e}_{2}^{*} + \tilde{e}_{1}^{*} > e_2^{*} + e_1^{*}$.
\end{proposition}

The Proposition states that the person works too much in total, working
too much on the earlier and too little on the later task.

The same mistake occurs when the person works on a single task
consisting of two or more subtasks, as long as each subtask is best done
in one block. If the subtasks have a natural sequence, so that one
subtask makes the subsequent subtask easier, then Elaine will work too
much on the earlier stages than on the later stages. For instance,
suppose that Elaine plans to read both the lecture notes and to finish a
problem set for the same class today. If she believes that the problem
set will be easier after reading the lecture notes, then she reads the
lecture notes first and consequently spends too much time on them.

While I assume that the order of the tasks is fixed exogenously in
Proposition \ref{prop:multi-tasking-single-day}, what is necessary for
the result to hold is that each task is done in one go, for which
switching cost are a sufficient condition.

\hypertarget{multi-day-all-or-nothing-tasks}{%
\section{Multi-Day All-or-Nothing
Tasks}\label{multi-day-all-or-nothing-tasks}}

In single-day all-or-nothing tasks, a person with increasing marginal
disutility always underestimates the costs of finishing the task. In
multi-day all-or-nothing tasks, she underestimates these costs at the
beginning of each day, but if she works long enough, she overestimates
them once her current marginal disutility is higher than the average
marginal disutility from completing the task. In this section, I develop
a continuous-time model to study how these fluctuations affect effort
choices and welfare.\footnote{In a working paper -- see Kaufmann (2020)
  -- I prove that this continuous-time solution is the limit of the
  discrete-time solution as \(T \to \infty\), which I take as a given
  here.}

\hypertarget{model-for-multi-day-all-or-nothing-tasks}{%
\subsection{Model for Multi-Day All-or-Nothing
Tasks}\label{model-for-multi-day-all-or-nothing-tasks}}

Consider a continuous-time setup where every time \(x \in [0, 1)\)
represents a different period. In each (continuous-time) period \(x\),
the person stops working once she perceives it as optimal doing so,
which determines the flow effort \(e_{x}\) at time \(x\). She incurs a
flow disutility of effort equal to \(D(e_{x})\). The person receives a
benefit of \(B\) at time \(x = 1\) if she has completed total effort
equal to or exceeding \(E\) and she receives no benefit otherwise.

So when does the projection-biased person stop working at time \(x\)?
Let us fix a period \(x\) from the continuum of periods. Then by this
definition, the person has a fraction \(1 - x\) of periods left to
complete the task, and I denote the amount of work left to do as
\(E_{x}\), i.e.~\(E_{x} = E - \int_{0}^{x} e_{t} dt\). If the
projection-biased person plans on completing the task, then she plans
completing it efficiently, working the same amount each day, which will
incur a cost of \(D\left( \frac{E_{x}}{1 - x}\right) (1 - x)\) in total.
After having worked for \(s\) hours on period \(x\), she perceives the
cost of completing the task efficiently as
\(G(x, s, E) := \tilde{D} \left( \frac{E_{x}}{1 - x}|s \right) (1 - x)\).
If at the start of the period when \(s = 0\) the perceived costs exceed
the benefits, then the person doesn't work at all, so that
\(e_{x} = 0\). If instead even after having worked for
\(s = \frac{E_{x}}{1 - x}\) hours the person still considers it worth
doing, then she stops working that period when
\(e_{x} = \frac{E_{x}}{1 - x}\). Otherwise she stops once the
\(G(x,s,E) = B\), since this implies that
\(G(x, s+\varepsilon, E) > B\), so she would never work more than \(s\)
hours. This shows that the following definition is in line with
Definition \ref{def:single-task-decision}:

\begin{definition}[Continuous-Time All-or-Nothing Task]
  Consider a projection-biased person working on a multi-period all-or-nothing task $(E, B)$ that requires effort $E$ to complete and pays monetary benefits $B$ if the task is completed by time $x = 1$ and pays $0$ otherwise. Let $E_x$ for $x \in [0, 1)$ be the effort that remains to be done at time $x$ in order to complete the task. Then $E_{x}$ is given by the ODE
\begin{align*}
    & \dot{E}_x = -e_{x} \\
    & e_x  =
    \begin{cases}
        0 \text{, if } G(x, 0, E_{x}) > B \\
        \frac{E_{x}}{1 - x} \text{, if } G(x, \frac{E_x}{1 - x}, E_{x}) < B \\
        e_x^{*} \text{ otherwise, with } G(x, e_x^{*}, E_{x}) = B
    \end{cases}
\end{align*}
with initial condition $E_0 = E$. A task is completed if $E_1 = 0$.
\end{definition}

Notice that because we normalize the time until the deadline to \(1\),
the total benefit \(B\) and effort \(E\) approximate the average
``daily'' benefit and effort in the discrete-time setting. That is, the
limit as \(T \to \infty\) holds the average daily benefits and effort
levels constant, not the totals.

\hypertarget{results-for-multi-day-all-or-nothing-tasks}{%
\subsection{Results for Multi-Day All-or-Nothing
Tasks}\label{results-for-multi-day-all-or-nothing-tasks}}

I first give a discrete-time example, the intuition of which carries
over to the continuous-time case. Consider Beth, a student who is
working on an all-or-nothing task with a deadline in \(T\) days. She has
an economics exam in \(100\) days and knows that she will receive a B in
her final if she does nothing but attend the required lectures. Getting
an A on the final is worth \(1250\) more than receiving a grade B. If
she studies \(5\) hours a day on average, Beth is sure to receive an A,
if she studies less, she is sure to receive a B.

Suppose that the daily disutility is quadratic: \(D(e) = \frac{e^2}{2}\)
so that \(D'(e) = e\). First, note that Beth at every moment either
plans to complete the task efficiently, or to not do the task at all.
After all, at any given moment she plans to do what an unbiased person
would do whose actual disutility was given by \(\tilde{D}(.|s)\). On the
first day, Beth therefore studies as long as she perceives it worthwhile
to study \(5\) hours every day. The disutility of studying \(5\) hours
per day is \(100 \cdot D(5) = 1250\), so an unbiased student would be
indifferent between studying and not studying. But Beth is
projection-biased, with \(\alpha = 0.5\). At the start of the first day
she underestimates the disutility of the task and starts studying. After
\(2.5\) hours of studying, her marginal disutility is \(D'(2.5)\), and
she perceives the disutility of working \(5\) hours on every future day
correctly:
\(\tilde{D}(5 | 2.5) = (1 - \alpha) D(5) + \alpha D'(2.5) \cdot 5 = D(5)\).
She therefore perceives the remaining disutility of studying \(5\) hours
every day almost correctly and soon stops working.\footnote{She still
  slightly underestimates it because she underestimates the disutility
  of the \(2.5\) hours of work she has to complete on the first day.}
When she stops, she believes, mistakenly, that she won't resume it the
next day. Yet, come the next day, she is rested and starts studying
again, planning to get an A, only to stop once more when she grows
sufficiently tired.

Every day, Beth thus either doesn't study at all, studies inefficiently
given how much work still remains to be done, or studies efficiently. It
is not difficult to see that if Beth doesn't study at all on day \(t\),
than she won't study on day \(t + 1\) or any later day either, and
therefore not get an A. Similarly, if she studies efficiently on day
\(t\), then she will study efficiently on all future days and thus get
an A. For instance, if after \(50\) days, Beth had only completed \(50\)
hours of studying, she would have to study \(9\) hours per day on the
remaining days, and she wouldn't start studying any longer.
Alternatively, if after \(75\) days Beth had completed \(300\) hours of
studying, she would have to work \(8\) hours a day for the remaining
\(25\) days to receive the full benefits worth \(1250\). She would work
\(8\) hours a day, since she would perceive this as worthwhile even
after \(8\) hours of work:
\(\tilde{D}(8|8) = 16 + 32 = 48 < 50 = \frac{1250}{25}\).

Proposition \ref{prop:multi-period-task-continuous} formally states that
for any average daily effort required, each of these two outcomes --
wasting effort on a task that won't be completed and working
inefficiently for a while on a task that is completed -- happens for
some average daily benefit. To state the result, I now define the time
\(\tau_0(E, B)\), which is (roughly) the fraction of days on which the
person doesn't work at all; and \(\tau_{F}(E, B)\) which is (roughly)
the fraction of days on which the person works efficiently on the task.
Formally:

\begin{definition}
Let $\tau_0(E_0, B_0) := \inf \{ \tau \in [0,1]: G(x, 0, E_{x}(E_0, B_0)) < B_0 \text{ } \forall x < 1 - \tau \}$. Thus for any $x < 1 - \tau_{0}$, we have $e_{x} > 0$.

Let $\tau_F(E_0, B_0) := \inf \{ \tau \in [0, 1]: G(x, \frac{E_{x}(E_0, B_0)}{1 - x}, E_{x}(E_0, B_0)) > B_0 \text{ } \forall x < 1 - \tau \}$. Thus for any $x < 1 - \tau_{F}$, the person does not work fully on the task (either not at all or only partially).
\end{definition}

We can now formally state the proposition.

\begin{proposition}
\label{prop:multi-period-task-continuous}
The disutility of effort is strictly convex, with $D''(\cdot) > d$ for some $d > 0$ and $\lim_{e \to \infty} D'(e) = \infty$. Consider a multi-period all-or-nothing task $(E_0, B_0)$. Then there exist $B_H(E_0) > B_C(E_0) > B_L(E_0) > 0$ such that:
\begin{itemize}
    \item if $B_{0} > B_H$, then the task is completed efficiently, i.e. $\tau_F = 1$.
    \item if $B_H> B_{0} > B_C$, then $\tau_F(E_0, B_0) \in (0, 1)$ and the task is completed.
    \item if $B_C > B_0 > B_L$, then $\tau_0(E_0, B_0) \in (0, 1)$ and the task is not completed.
    \item if $B_L > B_0$, then no effort is spent on the task, i.e. $\tau_0 = 1$.
\end{itemize}
with $\tau_{0}$ continuous and decreasing in $B_0$, and $\tau_{F}$ continuous and increasing in $B_0$.

Moreover $\lim_{B_0 \to B_C^{-}} u_0(E_0, B_0) \leq - D(E_0)$.
\end{proposition}

The proposition states that for any an all-or-nothing task
\((E_0, B_0)\) there exist thresholds \(B_H > B_C > B_L > 0\) depending
on \(E_0\) such that:

\begin{enumerate}
  \item Beth never starts the task if $B_0 < B_{L}$;
  \item Beth spends some days working on the task but eventually gives up if $B_0 \in (B_{L}, B_{C})$;
  \item Beth spends early days working inefficiently on the task but eventually finishes the remainder efficiently if $B_0 \in (B_{C}, B_{H})$;
  \item Beth works on the task efficiently from the start and finishes it if $B_0 > B_{H}$
\end{enumerate}

Moreover, as \(B_0 \in (B_{L}, B_{C})\) increases, Beth works a larger
fraction of days on the task and wastes more time on a task she doesn't
finish. As \(B_0\) approaches \(B_{C}\) from below, Beth ends up doing
almost all the work, yet fails to complete the task, hence she incurs
almost the full cost of completing the task. When \(B_0\) increases on
\((B_{C}, B_{H})\), Beth works more earlier, which makes her better off,
as she completes the task more efficiently.

Notice that the large cost from repeatedly working and stopping can
occur for arbitrarily small biases, although the range of payments for
which inefficient work happens becomes smaller the less biased a person
is. Thus the likelihood of mistakes decreases, but the range of possible
costs does not.

\hypertarget{productivity-and-effort-allocation}{%
\section{\texorpdfstring{Productivity and Effort Allocation
\label{sec:when-to-work}}{Productivity and Effort Allocation }}\label{productivity-and-effort-allocation}}

In previous sections, I assumed that the benefits are not projected,
since the value of money does not fluctuate much with tiredness. In this
section, I relax this assumption indirectly by considering situations
where a total output has to be produced, but the productivity (output
per hour of work) is either increasing or decreasing across periods. In
such a situation, the benefit of working consists in having less work in
the future.

\hypertarget{productivity-and-time-discounting}{%
\subsection{Productivity and Time
Discounting}\label{productivity-and-time-discounting}}

Let us consider a setting where people have different productivities on
different days, and let us start with a warm-up example. Doris has to
complete an assignment by tomorrow night that requires \(E\) hours. Her
productivity \(p\) on day \(1\) is twice as high as on day 2, because a
friend has offered to give feedback at the end of day 1. Thus every hour
of work exerted on day \(1\) leads to \(p = 2\) hours worth of output,
so she has to choose \(e_{1}\) and \(e_{2}\) s.t.
\(2 \cdot e_{1} + e_{2} = E\). On the first day after having worked for
\(s\) hours, she plans to stop after completing \(e_1(s)\) hours today,
given by

\begin{align}
    & \frac{\tilde{D}'(e_1(s) | s)}{p} = \tilde{D}'(e_2(s)|s) \notag \\
    \iff & \frac{\tilde{D}'(e_1(s) | s)}{p} - \tilde{D}'(e_2(s)|s) = 0 \notag \\
    \iff & (1 - \alpha) \frac{1}{p} \cdot D'(e_1(s)) + \alpha \frac{1}{p} \cdot D'(s) - (1 - \alpha) D'(e_2(s)) + \alpha D'(s) = 0 \notag \\
    \iff & \frac{1}{p} \cdot D'(e_1(s)) - D'(e_2(s)) = -\frac{\alpha}{1 - \alpha} D'(s) (\frac{1}{p} - 1) \notag \\
\end{align}

which shows that \(e_1(s) > e_1^{*}\), since
\(\frac{1}{p} = \frac{1}{2} < 1\) and \(D'(s) > 0\): the LHS has to
increase as \(\alpha\) increases, which requires that either
\(e_{1}(s)\) increases or \(e_{2}(s)\) decreases. But if \(e_{2}(s)\)
decreases, then \(e_{1}(s)\) increases and similarly if \(e_1(s)\)
increases then \(e_2(s)\) decreases given the constraint on total
effort. She stops working when her current perceived plan is equal to
(or less than) what she has done, that is when \(e_1(s) = s\).
Substituting this into the above equations we get:

\begin{equation*}
    D'(\tilde{e}_1) \left( \frac{1}{p} + \frac{\alpha}{1 - \alpha} (\frac{1}{p} - 1) \right) = D'(\tilde{e}_2) \iff \frac{D'(\tilde{e}_1)}{p} \frac{1 - p \alpha}{1 - \alpha} = D'(\tilde{e}_2)
\end{equation*}

and so she acts as if her productivity was
\(\frac{1 - \alpha}{1 - p \alpha} p > p\), so that she works too much on
day 1. This result is a special case of the proposition
\ref{prop:minimize-productivity-effort}.

\begin{proposition}
\label{prop:minimize-productivity-effort}
    Let $D(.)$ be a strictly convex function with $D'(\cdot) \to \infty$ and $D'(0)$, and let $\alpha \in [0, 1)$. Consider a person who works momentary each of $T$ periods under the constraint $E = \sum_{t = 1}^T p_t \cdot e_t$, where $p_t$ is her (known, exogenously given) productivity on day $t$. Denote by $\tilde{E}_t^{*}$ the total amount of work done by the beginning of day $t$ and by $E_{t}^{*}$ the optimal total amount of work done at the beginning of day $t$.
\begin{itemize}
  \item If $p_t$ is strictly increasing, then $\tilde{E}_t^{*} \leq E_t^{*}$ $\forall t > 1$.
  \item If $p_t$ is strictly decreasing, then $\tilde{E}_t^{*} \geq E_t^{*}$ $\forall t > 1$.
\end{itemize}
    Let $\tilde{e}_{t+1|t}^{*}$ denote the amount of work the person plans, at the end of day $t$, to exert on day $t + 1$. Then $\tilde{e}_{t+1|t}^{*} \leq \tilde{e}_{t+1}^{*}$, that is she works less on day $t + 1$ than predicted at the end of day $t$. When productivity is increasing, this change of plan moves her further away from optimal effort that day; when productivity is decreasing, this change of plan moves her closer to optimal effort that day.
\end{proposition}

Here is an example of proposition
\ref{prop:minimize-productivity-effort} in action. Betsy has to complete
an assignment that would take her \(E = 18\) hours of work if each day
she was as productive as she is today. Fortunately for her, she has
lectures tomorrow and office hours the day after that: during lectures
and office hours she will learn shortcuts for completing the problems on
the assignment. Specifically, every hour of work done tomorrow is worth
\(1.5\) hours of work today, while every hour of work done in two days
is worth \(1.5\) hours of work tomorrow.

Obviously, she should work less today than tomorrow, since she will
become more efficient at solving questions. Suppose that her disutility
is quadratic. The optimal effort levels should satisfy the first order
conditions
\(D'(e_1^{*}) = \frac{1}{p} D'(e_2^{*}) = \frac{1}{p^2} D'(e_3^{*})\),
which leads to \(e_1^{*} \approx 2.2\), \(e_2^{*} \approx 3.33\), and
\(e_3^{*} \approx 4.95\). On day \(1\), Betsy instead solves her
perceived first order conditions, which we can derive as in the
\(2\)-day case to be
\(D'(\tilde{e}_1) = \frac{1 - \alpha}{1 - \alpha \frac{1}{p}} \frac{1}{p} D'(\tilde{e}_{2|1}) = \frac{1 - \alpha}{1 - \alpha \frac{1}{p^2}} \frac{1}{p^2} D'(\tilde{e}_{3|1})\),
where \(\tilde{e}_{i|1}\) indicates that it is the effort Betsy
perceives to be optimal at the end of day \(1\). These are given by
\(\tilde{e}_1 \approx 1.52\), \(\tilde{e}_{2|1} \approx 3.03\), and
\(\tilde{e}_{3|1} \approx 5.29\).

Yet, on day \(2\) she will not do what she thought she would do. She
solves her new perceived first order condition, which is now exactly as
in the \(2\)-day case, taking into account that she worked roughly
\(1.52\) hours on day \(1\):
\(D'(\tilde{e}_2) = \frac{1 - \alpha}{1 - \alpha \frac{1}{p}} \frac{1}{p} D'(\tilde{e}_3)\).
Solving this, we find that \(\tilde{e}_2 \approx 2.75\) and that
\(\tilde{e_3} \approx 5.50\). Betsy was already planning to work less
than she should, planning to do \(3.03\) instead of \(3.33\), yet she
ends up doing even less, namely \(2.75\). Thus, Betsy postpones too much
work, and at the end of day 1 she thinks that she will have done more by
the end of day \(2\) than will be the case. The reason is that the more
tired she is, the more Betsy wants to do effort on a productive day --
which leads to postponing more work in this setting. Betsy correctly
understands that doing \(1\) hour less of work requires her to do \(40\)
minutes more work tomorrow. Thus she saves \(20\) minutes, which she
perceives as more unpleasant the more unpleasant effort is right now.
Therefore she is willing to delay more work until tomorrow to take
advantage of her higher productivity. Since tomorrow she will work more,
she will be more tired at the end of the day when she decides to stop,
and therefore she will want to delay more at the end of day \(2\) than
at the end of day \(1\) and stops working earlier than anticipated.

Since a person with constant productivity and exponential discounting
with discount factor \(\delta\) solves an identical problem to a person
whose whose productivity increases by \(\frac{1}{\delta}\) each period,
the next corollary follows immediately.

\begin{corollary}
    Let $D(.)$ be a strictly convex function with $D'(\cdot) \to \infty$ and $D'(0)$, and let $\alpha \in [0, 1)$. Consider a person who works momentary each of $T$ periods under the constraint $E = \sum_{t = 1}^T e_t$ and discounts disutility exponentially with discount factor $\delta$. Then $\tilde{E}_t^{*} \leq E_t^{*}$ $\forall t > 1$, and $\tilde{e}_{t+1|t}^{*} \geq \tilde{e}_{t+1}^{*}$.
\end{corollary}

Loewenstein, O'Donoghue, and Rabin (2003) highlighted the potential for
projection bias to cause time inconsistent plans under habit formation,
that is in a case where the utility from consumption decreases the more
one has recently consumed. What my results highlight though is that
projection bias can cause time-inconsistent behavior much more
generally, in fact whenever there are incentives towards unequal effort
over time. Moreover, it shows that the departure from earlier plans can
improve choices when effort is decreasing.

It is important to note however that projection-bias leads to
time-inconsistency only if the states in which choices are made are
different, and that they are not driven by time-inconsistent preferences
but by a failure to predict how different future preferences will be
from current preferences. This differs from temptation models such as
Gul and Pesendorfer (2001) or models of present bias (Laibson (1997);
O'Donoghue and Rabin (1999)), both of which assume that the actual
preferences over future outcomes depend on the immediacy of the choices.
Therefore projection bias does not preferentially lead to present focus
(Ericson and Laibson (2019)), although it leads to overestimating how
much one will work, if the final decision is made when people are the
most tired. Thus committed high-effort choices may be the mistake,
rather than the failure to implement them. Nonetheless, through
magnifying pre-existing present bias or present focus, projection bias
is more likely to magnify such behaviors. Since this magnification is
larger the more tired people are at the time of making committed
choices, this may bias estimates of time preferences when comparing
across populations or across times, with more tired populations
appearing as more impatient -- although careful laboratory designs can
avoid such concerns by keeping (expected) tiredness constant across
choice elicitations (see Fedyk (2018), Augenblick and Rabin (2019), Le
Yaouanq and Schwardmann (2019)).

\hypertarget{discussion-and-conclusion}{%
\section{Discussion and Conclusion}\label{discussion-and-conclusion}}

Throughout this paper, I highlighted how projection bias causes will
turn changing tiredness into changing plans, which leads to inefficient
task management. I made two major assumptions on the instantaneous
disutility. First I assumed that a person either works or doesn't work,
ruling out intensity of effort. Second I assumed that the instantaneous
disutility depends only on total time a person has worked so far, ruling
out breaks and rest during a day. The major assumptions that I made
about projection bias were that people make momentary decisions and that
they are naive so that they never realize that they do not execute their
earlier plans.

Adding intensity and rest to the model will help to integrate projection
bias better into applied settings, which is part of ongoing work.
Conceptually, more limiting is the assumption that people learn
\emph{nothing} from their repeated fluctuations. Without such naiveté,
plans will not fluctuate as often nor as much. While I believe that
naiveté is more appropriate than is often assumed, it seems clear that
people sometimes do display a sense of meta-sophistication: they realize
that they repeatedly fall short of their own expectations, that they
behave inconsistently, yet without a clearly articulated cause for this
behavior. The more challenging question, and one that is relevant for
for all types of misperceptions, is thus what people learn or do when we
neither assume that people \emph{must} learn nor that they \emph{cannot}
learn.

\pagebreak

\appendix

\section{Proofs}

\subsection{Proofs for Section 3}

Proof of proposition \ref{prop:optimal-effort}.

\begin{proof}
Remember that the biased person stops at $\tilde{e}^{*} = \inf \{s: B(e) - \tilde{D}(e|s) < B(s) - \tilde{D}(s|s) \text{, } \forall e > s\}$. It is therefore enough to show that when $s < e^{*}$, the projection-biased person perceives it as better to continue working a little more, and that if $s > e^{*}$, she perceives it as strictly optimal to stop right away, since then $\tilde{e}^{*} = e^{*}$.

Notice then that for $s < e^{*}$, we have $\tilde{D}'(s|s) = D'(s) < D'(e^{*}) = B'(e^{*}) \leq B'(s)$, so that the current perceived marginal disutility is strictly lower than the current marginal benefit. Thus even the biased person perceives it as strictly optimal to work (at least) a little more. Hence $\tilde{e}^{*} \geq e^{*}$. But similarly, when $s > e^{*}$, we have that $\tilde{D}'(s|s) = D'(s) > D'(e^{*}) = B'(e^{*}) \geq B'(s)$, and moreover, even the projection-biased person realizes that the marginal disutility will only increase and the marginal benefit decrease. Hence it is strictly better to stop right now rather than work more, so that $\tilde{e}^{*} \leq e^{*}$. Together, these imply that $\tilde{e}^{*} = e^{*}$.
\end{proof}

Proof of proposition \ref{prop:multi-tasking-single-day}:

\begin{proof}
The person -- by assumption -- first works on the first task, and then on the second task. After having worked on the first task for a time $s$, she plans on working $\tilde{e}_{1}^{*}(s)$ on task 1 and $\tilde{e}_{2}^{*}(s)$ on task 2 given by the following first order conditions:

\begin{equation*}
    \tilde{D}'(\tilde{e}_1^{*}(s) + \tilde{e}_2^{*}(s) | s) = B_1^{'}(\tilde{e}_1^{*}(s)) = B_2^{'}(\tilde{e}_2^{*})
\end{equation*}

Then she switches from working on task 1 to working on task 2 at time $\tilde{e}_1^{*}$, at which time she has to perceive her current marginal benefit from task 1 to be equal to her \emph{perceived} future marginal disutility at the end of the period:\footnote{The argument is similar to that in proposition \ref{prop:optimal-effort}: for any $s$ before the time where the perceived first order condition holds, the person considers working strictly more as strictly better. For any $s$ after that time, they perceive it as strictly worse. Hence they stop when the first order condition, as perceived in that moment, holds.}:
\begin{equation*}
    \tilde{D}'(\tilde{e}_1^{*} + \tilde{e}_{2|1}^{*} | \tilde{e}_1^{*}) = B_1^{'}(\tilde{e}_1^{*}) = B_2^{'}(\tilde{e}_{2|1}^{*})
\end{equation*}
where $\tilde{e}_{2|1}^{*}$ is the amount she plans to work on task $2$ at the time when she switches. Note that when $B_1^{'}(\tilde{e}_1^{*}) = B_2^{'}(\tilde{e}_{2|1}^{*})$ then the following are equivalent: $\tilde{e}_1^{*} > e_1^{*}$; $\tilde{e}_{2|1}^{*} > e_2^{*}$; and $\tilde{e}_1^{*} + \tilde{e}_{2|1}^{*} > e_1^{*} + e_2^{*}$, since the $B_i^{*}$ are strictly concave and disutilities are strictly convex. That is, if she plans to work more on the first task than is optimal, then she also plans to work more on the second task than is optimal and vice versa; and hence both of these imply, and are implied by, her planning to work more in total than is optimal.

We can now show that $\tilde{e}_1^{*} > e_1^{*}$. Suppose not. Then she plans to work less in total than is optimal and we get:
\begin{align*}
    \tilde{D}'(\tilde{e}_1^{*} + \tilde{e}_{2|1}^{*} | \tilde{e}_1^{*}) & < \tilde{D}'(\tilde{e}_1^{*} + \tilde{e}_{2|1}^{*} | \tilde{e}_1^{*} + \tilde{e}_{2|1}^{*}) \\
    & = D'(\tilde{e}_1^{*} + \tilde{e}_{2|1}^{*}) \\
    & \leq D'(e_1^{*} + e_2^{*}) \text{, since total optimal effort is larger then total planned effort} \\
    & = B'(e_1^{*}) \\
    & \leq B'(\tilde{e}_1^{*}) \text{, since } \tilde{e}_{1}^{*} \leq e_{1}^{*}
\end{align*}
which shows that it does not satisfy the first order condition. Thus $\tilde{e}_1^{*} > e_1^{*}$.

Once she switches, she keeps working on the second task. While she wants to reduce $e_1$, she can no longer do so, and thus takes $\tilde{e}_1^{*}$ as a given. Thus she now simply solves the first order condition
\begin{equation*}
    \tilde{D}'(\tilde{e}_1^{*} + \tilde{e}_2^{*}(s) | s) = B_2^{'}(\tilde{e}_2^{*}(s))
\end{equation*}
and as before, she stops once this holds her current $s$ equal to final total effort, $s = \tilde{e}_1^{*} + \tilde{e}_2^{*}$:
\begin{align*}
    & \tilde{D}'(\tilde{e}_1^{*} + \tilde{e}_2^{*} | \tilde{e}_1^{*} + \tilde{e}_2^{*}) = B_2^{'}(\tilde{e}_2^{*}) \\
    & \iff D'(\tilde{e}_1^{*} + \tilde{e}_2^{*}) = B'(\tilde{e}_2^{*}) \\
    & \iff D'(\tilde{e}_1^{*} + \tilde{e}_2^{*}) - B'(\tilde{e}_2^{*}) = 0 \\
    & \iff D'(e_1^{*} + x + \tilde{e}_2^{*}) - B'(\tilde{e}_2^{*}) = 0 \text{ for } x = \tilde{e}_{1}^{*} - e_{1}^{*} > 0
\end{align*}
Since $B'(\tilde{e}_{2}^{*}) = D'(e_1^{*} + e_2^{*} + x) > D'(e_1^{*} + e_2^{*}) = B'(e_2^{*})$, we have that $\tilde{e}_2^{*} < e_2^{*}$. Therefore, $D'(\tilde{e}_1^{*} + \tilde{e}_2^{*}) =  B'(\tilde{e}_2^{*}) > B'(e_2^{*}) = D'(e_1^{*} + e_2^{*})$, so that $\tilde{e}_1^{*} + \tilde{e}_2^{*} > e_1^{*} + e_2^{*}$ and we are done.
\end{proof}

Here is the proof of proposition \ref{prop:convex-single-day}.

\begin{proof}
Let $\tilde{R}(E|s) := \tilde{D}(E|s) - \tilde{D}(s|s)$, the perceived remaining disutility of completing the task after $s$ hours of work have already been completed. Before proving the main result, we have to deal with a technicality: the perceived optimal effort $\tilde{e}^{*}(s)$ at time $s$ is not necessarily unique as assumed in Section \ref{sec:setup}, since the person is indifferent between completing the task and stopping right away when $\tilde{R}(E|s) = B$. Thus in principle, the final outcome may be ambiguous, depending which way the indifference is broken. However, generically indifference will not matter. Let $s_{0}$ be the minimum $s$ such that $\tilde{R}(E|s_{0}) = B$. Such a minimum exists, since $\tilde{R}(E|\cdot)$ is continuous and hence the infimum is a minimum. But then, since $\tilde{R}(E|s)$ is increasing towards $B$, this means that generically the derivative of $\tilde{R}(E|\cdot)$ with respect to $s$ is strictly positive at $s_{0}$, so that $\tilde{R}(E|s_{0} + \varepsilon) > B$, and hence the person never works past $s_{0}$. Hence, the indifference at $s_{0}$ does not affect outcomes and the person stops at $s_{0}$. In what follows, I therefore assume that the person works as long as $\tilde{R}(E|s) < B$ and never works when $\tilde{R}(E|s) > B$, ignoring what happens at the non-generic points where $\tilde{R}(E|s) = B$.

The first part of the proposition claims that for all $E > 0$, if $B \in (\tilde{D}(E|0), D(E))$, then the person starts working on the task even though the task is not worth doing. Notice that $\tilde{D}(E|0) = (1 - \alpha) D(E) + \alpha D'(0) \cdot E < D(E)$, since $D(E) = \int_0^E D'(s) ds > \int_0^E D'(0) ds = D'(0) \cdot E$, as $D(\cdot)$ is strictly convex and $E > 0$. Thus for $B \in (\tilde{D}(E|0), D(E))$, since $\tilde{R}(E|0) = \tilde{D}(E|0)$ and since $\tilde{R}(E|s)$ is continuous in $s$, we have that $\tilde{R}(E|\varepsilon) < B$ for some sufficiently small $\varepsilon > 0$, so that the person will work at least for a time $\varepsilon$. This proves the first part of the proposition.

Let $\tilde{R}_{max} := \max_{s} \tilde{R}(E|s)$ the worst perceived remaining disutility for the biased person. Of course, for an unbiased person, the worst remaining disutility is always at the start when the most work remains to be done, but this won't necessarily hold for projection-biased people. The more she works, the more tired she gets, which makes her perceive the remaining work as worse than she perceived it earlier. Notice that if $\tilde{R}_{max}(E) < B$, then the remaining task is always perceived worth doing and therefore is completed. If $\tilde{R}_{max}(E) > B$, then the task is definitely not completed, since at some point the person perceives it not worth doing. Finally, if $\tilde{R}(E|0) < B$ and $\tilde{R}_{max}(E) > B$, then the person starts the task, but does not complete it.

Let $\mathcal{E} := \{E \geq 0: \tilde{R}_{max}(E) > \tilde{R}(E|0) \}$, the set of tasks for which the worst perception of the task happens after exerting some effort. I will show that $\mathcal{E} = (E_H, \infty)$ for some finite $E_H > 0$, which proves that if $E > E_H$, then we can pick $B$ in the non-empty interval $(\tilde{R}(E|0), \tilde{R}_{max}(E))$ and the person starts the task but fails to complete it. Moreover, I will show that if $E < E_H$, then $\tilde{R}(E|0) > \tilde{R}(E|s)$ $\forall s \in (0, E]$, which means that, for such tasks, if the person starts the task, she also completes it. These results then establish both the second and third part of the proposition.

First, let us show that $\mathcal{E}$ is not the empty set. Pick some $s > 0$ such that $D'(s) > 0$. Notice that $\tilde{R}(E|s) - \tilde{R}(E|0) = \tilde{D}(E|s) - \tilde{D}(s|s) - \tilde{D}(E|0) = (1 - \alpha) ( D(E) - D(s) - D(E) ) + \alpha E (D'(s) - D'(0)) + \alpha s D'(s) = - D(s) (1 - \alpha)  + \alpha s D'(s) + \alpha E (D'(s) - D'(0))$. Since $D'(s) - D'(0) > 0$, this expression becomes positive for sufficiently large $E$, say for $E > \bar{E}$, so that $\tilde{R}(E|s) - \tilde{R}(E|0) > 0$ for all $E > \bar{E}$. Thus $\mathcal{E}$ is not empty.

Further, from $\tilde{R}(E|s) - \tilde{R}(E|0) = - D(s) (1 - \alpha) + \alpha s D'(s) + \alpha E (D'(s) - D'(0))$ we immediately see that this expression is strictly increasing in $E$. Hence if $\tilde{R}_{max}(E) -  \tilde{R}(E|0) > 0$ and $E' > E$, then there is some $s$ such that $\tilde{R}(E|s) - \tilde{R}(E|0) > 0$ by definition of $\tilde{R}_{max}(E)$. Thus $\tilde{R}(E'|s) - \tilde{R}(E'|0) > 0$ and thus $\tilde{R}_{max}(E') - \tilde{R}(E'|0) > 0$. Therefore if $E \in \mathcal{E}$, then $E' \in \mathcal{E}$. Let $E_H = \inf \mathcal{E}$. Then if $E > E_H$, by definition of $E_H$, there is some $E' \in (E_H, E)$ s.t. $E' \in \mathcal{E}$. Therefore all $E > E_{H}$ are in $\mathcal{E}$. 

Moreover, $E_H \notin \mathcal{E}$, since either $E_H = 0$ (in which case it is obvious) or $E_H > 0$. If $E_H > 0$ and $E_H \in \mathcal{E}$, then $\tilde{R}(E_H|s) > \tilde{R}(E_H|0)$ for some $s > 0$, and thus $\tilde{R}(E_H - \varepsilon|s) > \tilde{R}(E_H - \varepsilon|0)$ for sufficiently small $\varepsilon$, which implies that $E_{H} - \varepsilon \in \mathcal{E}$. This contradicts the definition of $E_H$ as $\inf \mathcal{E}$. 

Finally, note that when $E < E_H$, we must have that $0 > \tilde{R}(E|s) - \tilde{R}(E|0)$ $\forall s > 0$. If not, then $\tilde{R}(E|s) - \tilde{R}(E|0) = 0$ for some $s$ and we know that the LHS strictly increases in $E$, which would imply that $E_H \in \mathcal{E}$. And thus we are done.
\end{proof}

\subsection{Proofs for Section 4}

Proof of proposition \ref{prop:multi-tasking-single-day}:

\begin{proof}
The person -- by assumption -- first works on the first task, and then on the second task. Consider the time $t^{*}$ that satisfies the following:
\begin{equation*}
    \tilde{D}'(t^{*} + \tilde{e}_{2|1}^{*} | t^{*}) = B_1^{'}(t^{*}) = B_2^{'}(\tilde{e}_{2|1}^{*})
\end{equation*}
where $\tilde{e}_{2|1}^{*}$ is the amount she plans to work on task $2$ at time $t^{*}$. Then I claim that $t^{*} = \tilde{e}_{1}^{*} = \inf \{s: \tilde{V}_1(e_1|s) < \tilde{V}_1(s|s) \text{,  } \forall e_1 > s\}$ -- that $t^{*}$ is the switching time. First, when $s < t^{*}$, then since $\tilde{D}(\cdot|s)$ is strictly increasing in $s$, we have that $\tilde{D}'(t^{*} + \tilde{e}_{2|1}^{*}|s ) < \tilde{D}'(t^{*} + \tilde{e}_{2|1}^{*} | t^{*}) = B_1^{'}(t^{*}) = B_{2}^{'}(\tilde{e}_{2|1}^{*})$. Hence she must plan to work more in total, and hence more on each task, than she does at $s = t^{*}$. Therefore she must plan to work more on both tasks (since she wants to equalize marginal benefits), therefore there is some $e > t^{*} > s$ such that $\tilde{V}_{1}(e|s) > \tilde{V}_{1}(s|s)$. Similarly, when $s > t^{*}$, we have that $\tilde{D}'(s + \tilde{e}_{2|1}^{*}|s ) > \tilde{D}'(t^{*} + \tilde{e}_{2|1}^{*} | t^{*}) = B_1^{'}(t^{*}) = B_{2}^{'}(\tilde{e}_{2{2|1}^{*}}) > B_1^{'}(s)$. But this means that the person plans less than $s + \tilde{e}_{2|1}^{*}$ total work, hence she has to plan to work less on at least one task. But she cannot work less on the first (she already spent time $s$), so it must be the second task, in which case the marginal benefit from the second task is strictly larger. Therefore, fixing the total amount of work, she should spend all the remaining time on the second task, and hence stops right now, i.e. $\tilde{V}_{1}(e|s) < \tilde{V}(s|s)$ for all $e > s$. This shows that $\tilde{e}_{1}^{*}$ satisfies:
\begin{equation*}
    \tilde{D}'(\tilde{e}_1^{*} + \tilde{e}_{2|1}^{*} | \tilde{e}_1^{*}) = B_1^{'}(\tilde{e}_1^{*}) = B_2^{'}(\tilde{e}_{2|1}^{*})
\end{equation*}
and $\tilde{e}_{2|1}^{*}$ is the amount she plans to work on task $2$ at the time of switching. 

Note that when $B_1^{'}(\tilde{e}_1^{*}) = B_2^{'}(\tilde{e}_{2|1}^{*})$ then the following are equivalent: $\tilde{e}_1^{*} > e_1^{*}$; $\tilde{e}_{2|1}^{*} > e_2^{*}$; and $\tilde{e}_1^{*} + \tilde{e}_{2|1}^{*} > e_1^{*} + e_2^{*}$, since the $B_i^{*}$ are strictly concave and disutilities are strictly convex. That is, if she plans to work more on the first task than is optimal, then she also plans to work more on the second task than is optimal and vice versa; and hence both of these imply, and are implied by, her planning to work more in total than is optimal.

We can now show that $\tilde{e}_1^{*} > e_1^{*}$. Suppose not. Then she plans to work less in total than is optimal and we get:
\begin{align*}
    \tilde{D}'(\tilde{e}_1^{*} + \tilde{e}_{2|1}^{*} | \tilde{e}_1^{*}) & < \tilde{D}'(\tilde{e}_1^{*} + \tilde{e}_{2|1}^{*} | \tilde{e}_1^{*} + \tilde{e}_{2|1}^{*}) \\
    & = D'(\tilde{e}_1^{*} + \tilde{e}_{2|1}^{*}) \\
    & \leq D'(e_1^{*} + e_2^{*}) \text{, since total optimal effort is larger then total planned effort} \\
    & = B'(e_1^{*}) \\
    & \leq B'(\tilde{e}_1^{*}) \text{, since } \tilde{e}_{1}^{*} \leq e_{1}^{*}
\end{align*}
which shows that it does not satisfy the first order condition. Thus $\tilde{e}_1^{*} > e_1^{*}$.

Once she switches, she keeps working on the second task. While she wants to reduce $e_1$, she can no longer do so, and thus takes $\tilde{e}_1^{*}$ as a given. Thus she now simply solves the first order condition
\begin{equation*}
    \tilde{D}'(\tilde{e}_1^{*} + \tilde{e}_2^{*}(s) | s) = B_2^{'}(\tilde{e}_2^{*}(s))
\end{equation*}
and as before, she stops once this holds her current $s$ equal to final total effort, $s = \tilde{e}_1^{*} + \tilde{e}_2^{*}$:
\begin{align*}
    & \tilde{D}'(\tilde{e}_1^{*} + \tilde{e}_2^{*} | \tilde{e}_1^{*} + \tilde{e}_2^{*}) = B_2^{'}(\tilde{e}_2^{*}) \\
    & \iff D'(\tilde{e}_1^{*} + \tilde{e}_2^{*}) = B'(\tilde{e}_2^{*}) \\
    & \iff D'(\tilde{e}_1^{*} + \tilde{e}_2^{*}) - B'(\tilde{e}_2^{*}) = 0 \\
    & \iff D'(e_1^{*} + x + \tilde{e}_2^{*}) - B'(\tilde{e}_2^{*}) = 0 \text{ for } x = \tilde{e}_{1}^{*} - e_{1}^{*} > 0
\end{align*}
Since $B'(\tilde{e}_{2}^{*}) = D'(e_1^{*} + e_2^{*} + x) > D'(e_1^{*} + e_2^{*}) = B'(e_2^{*})$, we have that $\tilde{e}_2^{*} < e_2^{*}$. Therefore, $D'(\tilde{e}_1^{*} + \tilde{e}_2^{*}) =  B'(\tilde{e}_2^{*}) > B'(e_2^{*}) = D'(e_1^{*} + e_2^{*})$, so that $\tilde{e}_1^{*} + \tilde{e}_2^{*} > e_1^{*} + e_2^{*}$ and we are done.
\end{proof}

\subsection{Proofs for Section 5}

Let me state again the initial value problem (henceforth IVP)
determining the behavior of a projection-biased person working on a task
requiring total effort \(E\) for a total benefit \(B\): \begin{align}
    & \dot{E}_{x} = - e_x \label{eq:IVP} \\
    & e_x  =
    \begin{cases}
        0 \text{, if } G(x, 0, E_{x}) > B \notag \\
        \frac{E_{x}}{1 - x} \text{, if } G(x, \frac{E_x}{1 - x}, E_{x}) < B \\
        e_x^{*} \text{ otherwise, with } G(x, e_x^{*}, E_{x}) = B 
    \end{cases}
\end{align}

where \(G(x, s, E) = (1 - x) \cdot \tilde{D}(\frac{E}{1 - x})\). When
\(G(x, e_{x}^{*}, E_{x}) = B\) holds, we have
\(e_{x}^{*} = f(x, E, B) := (D')^{-1}\left(\frac{B - (1 - \alpha) (1 - x) D(E)}{\alpha E }\right)\),
with \(f(\cdot)\) existing by strict convexity of \(D(\cdot)\).

As a reminder, we have the following definitions of \(\tau_{0}\) and
\(\tau_{F}\): \begin{align*}
  \tau_0(E_0, B_0) := \inf \{ \tau \in [0,1]: G(x, 0, E_{x}(E_0, B_0)) < B_0 \text{ } \forall x < 1 - \tau \} \\
  \tau_F(E_0, B_0) := \inf \{ \tau \in [0, 1]: G(x, \frac{E_{x}(E_0, B_0)}{1 - x}, E_{x}(E_0, B_0)) > B_0 \text{ } \forall x < 1 - \tau \}
\end{align*} Intuitively (but not quite), \(\tau_{0}\) is how close to
the deadline the person stops working and doesn't resume again; and
\(\tau_{F}\) is how close to the deadline the person starts working on
the task efficiently.

With this, let us first prove that \(E_x(E_0, B_0)\) for \(x < 1\) is
Lipschitz continuous in a neighborhood of \((E_0, B_0, x)\) and that it
is increasing in \(E_0\). I will use the following theorem (from
\url{https://www.math.washington.edu/~burke/crs/555/555_notes/continuity.pdf})
to prove continuity.

\begin{theorem}
\label{theorem:continuous-ode}
Consider the initial value problem 
\begin{equation*}
    x' = f(t, x, \mu), \text{ } x(t_0) = y
\end{equation*}

where $x'$ is the derivative of $x(t)$ with respect to time. If $f$ is continuous in $t$, $x$, $\mu$ and Lipschitz in $x$ with Lipschitz constant independent of $t$ and $\mu$, then $x(t, \mu, y)$ is continuous in $(t, \mu, y)$ jointly.
\end{theorem}

Then we get continuity as follows:

\begin{lemma}
\label{lemma:continuous-time-continuous-in-parameters}

Suppose that $D''(x) > d$ for some $d > 0$. The solution $E_x(E, B)$ to the continuous-time problem restricted to $x \in [0, 1- \varepsilon]$ with $\varepsilon > 0$ exists and is Lipschitz continuous in $x$, $E$, and $B$, on $[0, 1 - \varepsilon] \times [\underline{E}, \bar{E}] \times [0, \infty]$, for some $\bar{E} > \underline{E} > 0$.
\end{lemma}
\begin{proof}

It is clear that $E_x \leq E_0$, so we can pick $\bar{E} > E_0$. We then show that, starting with $E_0 \in [\underline{E}, \bar{E}]$, we will not fall below $\underline{E}$ before time $x$. Given that the maximum instantaneous effort is given by $\frac{E_x}{1 - x}$ it is not hard to see that at most a fraction $x$ of the total effort will be completed by time $x$ (the efficient amount, conditional on trying to complete the task).\footnote{This statement is proved in Lemma \ref{lemma:increasing-in-B}} Thus if $E_0 \geq \frac{1}{\varepsilon} \underline{E}$, then $E_x$ will be larger than $\underline{E}$ for all $x \leq 1 - \varepsilon$. 

Given theorem \ref{theorem:continuous-ode} and our IVP (\ref{eq:IVP}), we only need to show that $e_x(x, E, B)$ is continuous in $t$, $E$, and $B$, and Lipschitz continuous in $E$ independent of $t$ and $B$. Notice that $G$ and $f$ are continuous functions, given that $x$ is bounded away from $1$ and $E$ is bounded away from $0$.

First, notice that when $G(x, 0, E) = B$, by definition of $G$ and $f$ we have that $f(x, E, B) = (D')^{-1}(D'(0)) = 0$, and similarly when $G(x, \frac{E}{1 - x}, E) = B$, we have that $f(x, E, B) = \frac{E}{1 - x}$. Thus $e(x, E, B)$ restricted to $A := \{(x, E, B): G(x, 0, E) \geq B\}$ is the constant $0$ function, $e(x, E, B)$ restricted to $B := \{(x, E, B): G(x, \frac{E}{1 - x}, E) \leq B\}$ is equal to $\frac{E}{1 - x}$, and $e(x, E, B)$ restricted to $C := \{(x, E, B): G(x, 0, E) \leq B \text{ and } G(x, \frac{E}{1 - x}, E) \geq B \}$ is equal to $f(x, E, B)$.

If we can show that $e(x, E, B)$ restricted to $\mathcal{A}$, $\mathcal{B}$, and $\mathcal{C}$ is Lipschitz in all parameters (which is stronger than what we need), then $e(x, E, B)$ is Lipschitz continuous in all parameters over the union of $\mathcal{A}$, $\mathcal{B}$, and $\mathcal{C}$. The reason is that all three regions are closed, and thus contain their limit points: Suppose we have two points $\vect{x} = (x, E, B)$ and $\vect{x'} = (x', E', B')$ and we want to show that $|e(x, E, B) - e(x', E', B')| < K(|x - x'| + |E - E'| + |B - B'|)$ for some $K$. First, if both points are in the same region, then this immediately holds, by the assumption that the function is Lipschitz in that region. Now suppose that the two points are in regions $\mathcal{A}$ and $\mathcal{C}$. These two regions share a common border. Thus there exists some point $\vect{x''} = (x'', E'', B'') = \kappa \cdot \vect{x} + (1 - \kappa) \cdot \vect{x'}$ on the line connecting the two points that belongs to both regions (this is the part that requires both $\mathcal{A}$ and $\mathcal{C}$ to be closed), so that $|e(x, E, B) - e(x', E', B')| = |e(x, E, B) - e(x'', E'', B'') + e(x'', E'', B'') - e(x', E', B')| \leq |e(x, E, B) - e(x'', E'', B'')| + |e(x'', E'', B'') - e(x', E', B')| < K (|x - x''| + |E - E''| + |B - B''| + |x'' - x'| + |E'' - E'| + |B'' - B'|) = K(|x - x'| + |E - E'| + |B - B'|)$, where $|x - x''| + |x'' - x'| = |x - x'|$ because the point $\vec{x''}$ lies between the two points (is a convex combination of) $\vect{x}$ and $\vect{x'}$. Thus the function is Lipschitz continuous over the union of $\mathcal{A}$ and $\mathcal{C}$, and by an exactly identical argument over the union of the three regions.

Restricting ourselves to $E_0 \in [\underline{E}, \bar{E}]$, it is clear that $e(x, E, B)$ is Lipschitz on $A$, where it is constant. It is equally clear that $e(x, E, B)$ is Lipschitz continuous on $B$ since (by assumption) we are only considering $x \leq 1 - \varepsilon$, that is $1 - x \geq \varepsilon$. 

Finally, $e(x, E, B)$ is Lipschitz continuous on $C$ if $f(\cdot)$ is. But $f(\cdot)$ is the inverse function of $D'(\cdot)$, so as long as the derivative of $D'(\cdot)$ is strictly bounded away from $0$ everywhere, $f(\cdot)$ is Lipschitz. This holds since we assume $D''(x) > d$ for some $d > 0$. Thus we have shown that $e(x, E, B)$ is Lipschitz continuous when $x \leq 1 - \varepsilon$, $E \in [\underline{E}, \bar{E}]$ and $B \geq 0$, for any $\varepsilon > 0$, $\underline{E} > 0$, $\bar{E} > 0$.
\end{proof}

\begin{lemma}
\label{lemma:absorbing}
If $G(x, 0, E_x) > B$, then $G(x', 0, E_{x'}) > B$ for all $x' \geq x$. Similarly, if $G(x, \frac{E_x}{1 - x}, E_x) < B$, then $G(x', \frac{E_{x'}}{1 - x'}, E_{x'}) < B$ for all $x' \geq x$.
\end{lemma}
\begin{proof}
  Suppose not. Then there exists $1 > x' > x$ such that $G(x', 0, E_{x'}) \leq B$. Note that by lemma~\ref{lemma:continuous-time-continuous-in-parameters}, $E_x$ is continuous on $[0, x' + \varepsilon]$ for sufficiently small $\varepsilon$, and because $G$ is continuous in all its arguments, we know that $G(x + \varepsilon_1, 0, E_{x + \varepsilon_1}) > B$ for sufficiently small $\varepsilon_1$. Thus for $x^{*} := \inf \{x' > x: G(x', 0, E_{x'}) \leq B \}$, we have $x^{*} > x$. Moreover, $G(y, 0, E_y) > B$ for all $x \leq y < x^{*}$ and therefore $e_y = 0$. Thus $E_{x^{*}} = E_x - \int_x^{x^{*}} e_y dy = E_x$. Hence we have that $G(x^{*}, 0, E_{x^{*}}) = G(x^{*}, 0, E_x) > G(x, 0, E_x) > B$, since $G$ is strictly increasing in $x$.\footnote{$G$ is strictly increasing in $x$ because $(1 - x) \cdot D(E/(1 - x)) = E \cdot 1/X \cdot D(X)$ where $X  = (1 - x)/E$. But $D(X)/X$ is the average disutility per unit of effort, which strictly increases for a strictly convex function $D(\cdot)$.} But then by continuity of $E_x$ and $G$, we have that $G(x^{*} + \varepsilon_2, 0, E_{x^{*} + \varepsilon_2}) > B$ for sufficiently small $\varepsilon_2$, which contradicts the definition of $x^{*}$.

A similar argument works for the second part of the lemma.
\end{proof}

\begin{lemma}
\label{lemma:increasing-in-E}
For a fixed $x < 1$ and $B_0 > 0$, $E_x(E_0, B_0)$ is strictly increasing in $E_0$.
\end{lemma}
\begin{proof}
Let $\Delta_x = E_x(E_0^{'}, B_0) - E_x(E_0, B_0)$ for some $E_0^{'} > E_0 > 0$. We need to show that $\Delta_x > 0$ for all $x < 1$.

Notice that $\Delta_0 = E_0^{'} - E_0 > 0$ and that $\frac{d \Delta}{dx} = - e_x^{'} + e_x$. Since $E_x$ is continuous, we have that $\Delta_x > 0$ for all $x < \varepsilon$ at least. Suppose that the claim is false, so that $x^{*} := \inf \{x: \Delta_{x^{*}} \leq 0 \}$ exists. Then $x^{*} \geq \varepsilon > 0$ and for all $x < x^{*}$ we have $\Delta_{x} > 0$. Thus $E_x^{'} > E_x$.

We then have check that for all possible cases of values for $e_{x}$, we can limit how large $e_{x'}$ is. If $G(x, 0, E_x) > B$, then $G(x, 0, E_x^{'})$, and hence $e_x = e_x^{'} = 0$. If $G(x, e_x, E_x) = B$, then $G(x, e_x, E_x^{'}) > B$, since $G$ is increasing in $E$, and therefore $e_x^{'} < e_x$ because $G$ is increasing in its second argument. Finally, if $G(x, \frac{E_x}{1 - x}, E_x) < B$, then $e_x = \frac{E_x}{1 - x}$ and $e_x^{'} \leq \frac{E_x^{'}}{1 - x}$.

Thus we see that in all cases $\frac{d \Delta_x}{dx}  = - e_{x}^{'} + e_{x} \geq \frac{- E_{x}^{'} + E_{x}}{1 - x} = - \frac{\Delta_x}{1 - x}$, and therefore for $x < x^{*}$ $\Delta_{x^{*}} = \Delta_x + \int_x^{x^{*}} \frac{d \Delta_x}{dx} dx \geq \Delta_x - \int_x^{x^{*}} \frac{\Delta_y}{1 - y} dy$. Let $\delta < \frac{1}{2} (1 - x^{*})$. We know by the definition of $x^{*}$ that $\Delta_x > 0$ for $x < x^{*}$. Pick $x \in [x^{*} - \delta, x^{*})$ that achieves the maximum of $\Delta_x$ in this interval, which exists since $\Delta_x$ is continuous. Then we have that $\Delta_{x^{*}} \geq \Delta_x - \int_x^{x^{*}} \frac{\Delta_y}{1 - y} dy \geq \Delta_x - \Delta_x \frac{\delta}{1 - x^{*}} > \frac{1}{2} \Delta_x > 0$ . Thus $\Delta_{x^{*}} > 0$ and therefore (by continuity) $\Delta_{x^{*} + \varepsilon} > 0$ for some small $\varepsilon > 0$, which contradicts the definition of $x^{*}$. Thus the claim is proved.
\end{proof}

\begin{lemma}
\label{lemma:increasing-in-B}
$E_x(E_0, B_0) \geq E_{x'}(E_{0}, B_{0}) \frac{1 - x}{1 - x'}$ for $1 > x > x' \geq 0$, and $E_{x}(E_{0}, B_{0})$ is decreasing in $B_{0}$.
\end{lemma}
\begin{proof}

  \begin{align*}
    \dot{E}_{x} \geq \frac{E_{x}}{1 - x} & \implies \frac{\dot{E}_{x}}{E_{x}} \geq - \frac{1}{1 - x} \\
                                         & \implies \frac{d}{dx} \log(E_{x}) \geq - \frac{1}{1 - x} \\
                                         & \implies \log(E_{x'}) - \log(E_{x}) \geq - \int_{x}^{x'} \frac{1}{1 - y} dy \\
                                         & \implies \log(\frac{E_{x'}}{E_{x}}) \geq \int_{x}^{x'} \frac{d}{dy} \log(1 - y) dy \\
                                         & \implies \log(\frac{E_{x'}}{E_{x}}) \geq \log \frac{1 - x}{1 - x'} \\
                                         & \implies \frac{E_{x'}}{E_{x}} \geq \frac{1 - x}{1 - x'}
  \end{align*}
  The proof that $E_x(E_0, B_0)$ is decreasing in $B_0$ is similar to the proof of lemma \ref{lemma:increasing-in-E}, and thus I omit it.
  % TODO: At some point write it out in full, if I have time.
\end{proof}

Now let us prove that \(\tau_0(E_0, B_0)\) and \(\tau_F(E_0, B_0)\) are
continuous in \(E_0\).

\begin{lemma}
\label{lemma:tau-and-utility-continuous}
Suppose $D'(E) \to \infty$ as $E \to \infty$ and $D''(\cdot) > d$ for some $d > 0$. Then $\tau_{0}(E,B)$ is increasing in $E$ and decreasing in $B$, and if $\tau_{0}(E,B) \in (0,1)$ it is continuous in $(E, B)$ in a neighborhood of $(E_0,B_{0})$. Similarly, $\tau_F(E, B)$ is decreasing in $E$ and increasing in $B$, and if $\tau_{F}(E,B) \in (0, 1)$ it is continuous in $(E,B)$ in a neighborhood of $(E_0,B_{0})$.
\end{lemma}

\begin{proof}
The proofs are essentially identical for $\tau_0$ and $\tau_F$, so I prove the first. We know

\begin{equation*}
    \tau_0(E_0, B_0) = \inf \{ \tau: 1 - \tau \in [0, 1] \text{ and } G(x, 0, E_{x}(E_0, B_0)) < B_0 \text{ } \forall x < 1 - \tau \} = \inf \Gamma_0
\end{equation*}

Notice that $1 \in \Gamma_0$, thus $\tau_0$ always exists. Then take $x > 1 - \tau_0$. Suppose $E_0^{'} > E_0$ and let $\tau_0 := \tau_0(E_0, B_0)$ and $\tau_0^{'} := \tau_0(E_0^{'}, B_0)$ and similarly for $\Gamma_0$ and $\Gamma_0^{'}$. Note that if $G(x, 0, E_x(E_0, B_0)) > B_{0}$ then, by lemma \ref{lemma:increasing-in-E}, $E_x^{'} \geq E_x$, and thus (since $G$ is increasing in its third argument) $G(x, 0, E_x(E_0, B_0)) > B_{0}$. Therefore for $\tau \notin \Gamma_0$, then there exists some $x < 1 - \tau$ with $G(x, 0, E_x) \geq B_{0}$ and therefore $G(x, 0, E_x^{'}) > B_{0}$ so that $\tau \notin \Gamma_0^{'}$. Hence $\Gamma_0^{'} \subset \Gamma_0$ and thus $\tau_0^{'} \geq \tau_0$.

Notice that as we increase $B_{0}$ to $B_{0}^{'}$, every $\tau$ in $\Gamma_{0}$ is necessarily also in $\Gamma_{0}^{'}$: if $G(x, 0, E_{x}(E_{0},B_{0})) < B_{0}$, then $G(x, 0, E_{x}(E_{0}, B_{0}^{'})) < B_{0}^{'}$, since $E_{x}$ weakly decreases in $B_{0}$, and thus $G(\cdot)$ decreases, while the RHS increases. Thus $\Gamma_{0} \subset \Gamma_{0}^{'}$, hence $\tau_{0}^{'} \leq \tau_{0}$.

Let us now show continuity for $\tau_{0} \in (0, 1)$. Fix $(E_{0}, B_{0})$, then $\tau_{0}$ is s.t. for $x < 1 - \tau_{0}$ we have $G(x, 0, E_{x}(E_{0}, B_{0})) < B_{0}$ and for every $\varepsilon > 0$ there is some $x \in (1 - \tau_{0}, 1 - \tau_{0} + \varepsilon)$ with $G(x, 0, E_{x}(E_{0}, B_{0})) \geq B_{0}$. Suppose by contradiction that $\tau_{0}$ is not continuous. Then there is some $\delta$ s.t. for every $\varepsilon_{i}$ we have $(E_{i}, B_{i})$ within $\varepsilon_{i}$ distance from $(E_{0}, B_{0})$ with either some $x_{i} \leq 1 - \tau_{0} - \delta$ and $G(x_{i}, 0, E_{x_{i}}(E, B)) \geq B$ so that $\tau_{0}^{'} \geq \tau_{0} + \delta$, or we have for all $x < 1 - \tau_{0} + \delta$ we have $G(x, 0, E_{x}(E, B)) < B$, so that $\tau_{0}^{'} \leq \tau_{0} - \delta$.

\textbf{Contradiction Case 1:} For $\varepsilon_{i} \to 0$, there is a sequence $(E_{i}, B_{i})$ s.t. $G(x, 0, E_{x}(E_{i}, B_{i})) < B_{i}$ for all $x < 1 - \tau_{0} + \delta$.

Since $E_{x}(E,B)$ is (Lipschitz) continuous in $(E, B)$ and $G(\cdot)$ in its arguments in the range observed, this converges uniformly for all $x < 1 - \tau_{0} + \delta$. Applying this to the closed range $x < 1 - \tau_{0} + 1/2 \delta$, we find that $G(x, 0, E_{x}(E_{0}, B_{0})) < B_{0}$ for all $x < 1 - \tau_{0} + 1/2 \delta$, which contradicts the value of $\tau_{0}$.

\textbf{Contradiction Case 2:} For $\varepsilon_{i} \to 0$, there is a sequence $(E_{i}, B_{i})$ and some $x_{i} \leq 1 - \tau_{0} - \delta$ s.t. $G(x, 0, E_{x}(E_{i}, B_{i})) \geq B_{i}$.

Given that the ranges are all finite, $(E_{i}, B_{i}, x_{i})$ converges to $(E_{0}, B_{0}, x)$, with $x \leq 1 - \tau_{0} - \delta$, with $G(x, 0, E_{x}(E_{0}, B_{0})) \geq B_{0}$. But this directly contradicts the definition of $\tau_{0}$, since this implied that $G(x, 0, E_{x}(E_{0}, B_{0})) < B$ for all $x < 1 - \tau_{0}$.

Hence $\tau_{0}(E, B)$ is continuous in $(E, B)$
\end{proof}

\begin{lemma}
\label{lemma:disutility-to-infinity}
Let $D$ be convex with $D'(e) \to \infty$ as $e \to \infty$. Then $\forall K > 0$, $\exists E$ s.t. $D(e) > K \cdot e$ $\forall e > E$. That is, $D(e)/e \to \infty$ as $e \to \infty$.
\end{lemma}
\begin{proof}
Since $D'(e) \to \infty$, pick $E$ s.t. $D'(\frac{E}{2}) > 2 \cdot K$. Then for $e > E$

\begin{align*}
    D(e) = \int_0^{e} D'(s) ds \geq \int_{E/2}^e D'(s) ds \geq \int_{E/2}^E 2 \cdot K ds \geq \frac{e}{2} 2 \cdot K = e \cdot K
\end{align*}
\end{proof}

Proof of proposition \ref{prop:multi-period-task-continuous}:

\begin{proof}
Let $B_L = (1 - \alpha) D(E_0) + \alpha D'(0) E_0$. Then $G(0, 0, E_0) = B_L$ and therefore if $B < B_L$ we have $G(0, 0, E_0) > B_L$ and hence by lemma \ref{lemma:absorbing} we know that $G(x, 0, E_x) > B_L$ for all $x \geq 0$. Hence $e_x = 0$ and $\tau_0(E_0, B_0) = 1$. Similarly, if $B_H = (1 - \alpha) D(E_0) + \alpha D'(E_0) E_0$, then $G(0, E_0, E_0) = B_H$. Hence if $B > B_H$ we have $G(0, E_0, E_0) < B$ and again by lemma \ref{lemma:absorbing} this holds for all $x \geq 0$ and thus $\tau_F = 1$ and $e_x = E_0$ (this last part in effect requires solving the same differential equation as we did in lemma \ref{lemma:increasing-in-B}, which I omit).

Moreover, note that if $B < B_H$ then we have that $G(0, E_0, E_0) > B$ and thus (by continuity of $E_x$ and $G$) we have that $G(x, \frac{E_x}{1 - x}, E_x) > B$ for all sufficiently small $x$. Therefore, $\tau_F < 1$. Similarly, if $B > B_L$ we have that $\tau_0 < 1$.

It is clear by lemma~\ref{lemma:absorbing} that if $\tau_0 > 0$ then $\tau_F = 0$ and if $\tau_F > 0$ then $\tau_0 = 0$: if $\tau_{0} > 0$, then there is some $x \in [1 - \tau_{0}, 1 - \tau_{0} + \varepsilon)$ such that $e_{x} = 0$, hence by lemma~\ref{lemma:absorbing} we have $e_{x'} = 0$ for all $x' > x$, so the person never works efficiently on the task. The other direction is similar. Let $B_{C,0} = \inf \{B: \tau_0(E_0, B) = 0 \}$ -- roughly the smallest $B$ for which there is some work done at all times $x$. Then because $\tau_0$ is decreasing in $B_0$ by lemma~\ref{lemma:tau-and-utility-continuous}, we know that if $B < B_{C, 0}$ then $\tau_0(E_0, B) > 0$, since if $\tau_0(E_0, B) = 0$, then $\tau_0(E_0, B') = 0$ for all $B' \geq B$, contradicting the definition of $B_{C, 0}$. Similarly we can define $B_{C, F} = \limsup \{B: \tau_F(E_0, B) = 0\}$ and show that if $B > B_{C, F}$ then $\tau_F > 0$.

% TODO: Check next paragraph

To finish the proof, we need to show that $B_{C, F} = B_{C, 0}$. Notice that if $B_0 \in [B_{C, 0}, B_{C, F}]$ we have that $\tau_0 = 0$ and $\tau_F = 0$. Therefore $G(x, e_x, E_x(E_0, B_0)) = B$ for all $x < 1$. Suppose that $B_{C, 0} < B_{C, F}$.\footnote{We cannot have $B_{C,0} > B_{C, F}$ since then $\tau_{F} > 0$ and $\tau_{0} > 0$ for all $x \in (B_{C,F}, B_{C,0})$, but both cannot happen jointly.} Let $e_{x,0}$ be the effort for $B_{C,0}$ and $e_{x,F}$ the effort for $B_{C,F}$. Then, since $G(x, e_{x, 0}, E_{x, 0}) = B_{C, 0} < B_{C, F} = G(x, e_{x, F}, E_{x, F})$ for all $x$, we must have that $e_{x, F} > e_{x, 0}$ or $E_{x, F} > E_{x, 0}$ for every $x$. Since $E_{0, C} = E_{0, F}$, by continuity of $E_x$ in $x$ and $G$ in $E$, we can pick $\varepsilon > 0$ such that for all $x < \varepsilon$, $G(x, e_{x,0}, E_{x, F})$ is arbitrarily close to $G(x, e_{x, 0}, E_{x, 0}) = B_{C, 0}$. Therefore the inequality holds only if $e_{x, F} > e_{x, 0}$ for $x < \varepsilon$, and thus $E_{x, F} < E_{x, 0}$. We can then show that $e_{x, F} > e_{x, 0}$ for all $x$. Suppose not, then we must have that $E_{x, F} > E_{x, 0}$ for some $x$ and therefore there exists a smallest $x^{*} > \varepsilon$ such that $E_{x^{*}, F} = E_{x^{*}, 0}$. But $e_{x, F} > e_{x, 0}$ for all $x < x^{*}$, therefore $E_{x^{*}, F} < E_{x^{*}, 0}$, which is a contradiction.

Thus we have shown that $E_{x, F} < E_{x, 0}$ and that $e_{x, F} > e_{x, 0}$ for all $x > 0$. Let $\delta = E_{\frac{1}{2}, 0} - E_{\frac{1}{2}, F} > 0$, then $E_{x, 0} - E_{x, F} \geq \delta$ for $x > \frac{1}{2}$ and therefore $E_{x, 0} \geq \delta > 0$ for all $x$. Therefore $D(\frac{E_{x,0}}{1 - x})(1 - x) \geq D(\frac{\delta}{1 - x}) \frac{1 - x}{\delta} \delta \to \infty$ as $x \to 1$ by lemma \ref{lemma:disutility-to-infinity}. But this means that $G(x, 0, E_{x, 0}) \to \infty$ and therefore that $G(x, 0, E_{x, 0}) > B$ as $x \to 1$, so that $\tau_0 > 0$. Therefore, we cannot have that $B_{C,0} < B_{C, F}$, so that $B_{C,0} = B_{C,F} = B_{C}$, and we are done.
\end{proof}

Now let us show that the utility is continuous and decreasing on
\((B_L, B_C)\) and continuous and increasing on \((B_C, B_H)\).

\begin{lemma}
\label{lemma:utility-bounds}
The utility $u_0(E_0, B_0) := - \int_0^1 D(e_x) dx$ is continuous and decreasing on $(B_L(E_0), B_C(E_0))$ and the utility $u_F(E_0, B_0) := B - \int_0^1 D(e_x) dx$ is continuous and increasing on $(B_C(E_0), B_H(E_0))$.
\end{lemma}
\begin{proof}
Notice that when $B \in (B_L, B_C)$ then we know that $\tau_0 \in (0, 1)$ and the task is not completed, hence the definition of the utility as $u_0$ is correct. Moreover $u_0 = \int_0^{1 - \tau_0} D(e_x) dx$. We can show that $\tau_0$ and $E_x$ are continuous and decreasing in $B_0$. Picking $B_0 < B_0^{'}$, we therefore have that $\tau_0^{'} < \tau_0$ and that for $x \leq 1 - \tau_0$ we have $G(x, e_x, E_x) = B_0 < B_0^{'} = G(x, e_x^{'}, E_x^{'})$. Since $E_x^{'} \leq E_x$ we therefore have that $e_x^{'} > e_x$ and therefore $u_0^{'} > \int_0^{1 - \tau_0} D(e_x^{'}) dx > \int_0^{1 - \tau_0} D(e_x) dx = u_0$. Moreover, if $B_0^{'}$ is close to $B_0$ then $E_x$ is close to $E_x^{'}$ by Lipschitz continuity and therefore $e_x^{'}$ and $e_x$ are close together, since $e_x$ is Lipschitz continuous in all the parameters as well (I haven't shown this in detail, but this is where I use the condition $D'(0) > 0$). Therefore the $u_0^{'}$ and $u_0$ are close.

Now suppose $B_0$, $B_0^{'} \in (B_C, B_H)$ then $\tau_F \in (0, 1)$. Let $B_0 < B_0^{'}$. We can show in a similar way as before that $\tau_F^{'} > \tau_F$ and that $e_x^{'} > e_x$ for $x \leq 1 - \tau_0^{'}$. Then notice that $\int_0^1 e_x = E_0 = \int_0^1 e_x^{'}$. Let $F(e) := \mathbb{P}(e_{x} \leq e) = \int_0^1 \mathbbm{1}(e_x \leq e) dx$ and $G(e) := \mathbb{P}(e_{x}^{'} \leq e) = \int_0^1 \mathbbm{1}(e_x^{'} \leq e) dx$, where the interpretation as probabilities is to help intuition, although it can be made formal by drawing $x$ uniformly from $[0,1)$. We want to show that $\int_{0}^{1} D(e) dF(e) > \int_{0}^{1} D(e) dG(e)$. By strict convexity of $D(\cdot)$, this holds if $F(\cdot)$ is a mean-preserving spread of $G(\cdot)$. Let $\bar{e}_{G} = e_{1 - \tau_F^{'}}$, be the effort the person exerts under $G(\cdot)$ once they work fully, with $B_{0}^{'}$. Then if $e < \bar{e}$, $G(e) < F(e)$, since $e_x^{'} > e_x$ for all $x \leq 1 - \tau_0^{'}$), which are the only $e$ that can below $\bar{e}$ under $G(\cdot)$, while $F(\cdot)$ can get contributions from $e > \bar{e}_{G}$. However $G(\bar{e}) = 1$ for $e > \bar{e}_{G}$, while $F(\bar{e}) \geq 1$, we have that $F(e) \geq G(e)$ for $e \geq \bar{e}$. Therefore $F$ is a mean-preserving spread of $G$, moving effort from above $\bar{e}$ to below, and thus the disutility for $e_x$ is higher than for $e_x^{'}$.\footnote{A rigorous proof of this would require to show that $e_{x}^{'} < \bar{e}_{G}$, so that the final effort exerted is the highest effort ever exerted.} Continuity follows again by noting that, until time $1 - \tau_0$, $e_x$ is Lipschitz continuous in all parameters, and thereafter it is constant. Therefore the utility is Lipschitz continuous.
\end{proof}

I will need the following lemma to prove the second part of the
proposition.

\begin{lemma}
\label{lemma:little-left-to-do-to}
Let $D$ be convex and such that $D'(e) \to \infty$ as $e \to \infty$. Fix $B$ and $\varepsilon > 0$. Let $e_{\varepsilon}$ be s.t. 
\begin{equation}
    D(e_{\varepsilon}) \cdot \varepsilon = B
    \label{eq:indifferent}
\end{equation}
Then $e_{\varepsilon} \cdot \varepsilon \to 0$ as $\varepsilon \to 0$.
\end{lemma}
\begin{proof}
First note that as $\varepsilon$ goes to $0$, $e_{\varepsilon}$ goes to $\infty$, since if it was bounded, then $D(e_{\varepsilon}) \cdot \varepsilon$ would go to $0$. By lemma \ref{lemma:disutility-to-infinity}, we know that $\frac{D(e_{\varepsilon})}{e_{\varepsilon}} \to \infty$. Dividing both sides of equation \ref{eq:indifferent} by $\varepsilon \cdot e_{\varepsilon}$ yields

\begin{align*}
    D(e_{\varepsilon})/e_{\varepsilon} = \frac{B}{e_{\varepsilon} \cdot \varepsilon} \iff  e_{\varepsilon} \cdot \varepsilon = \frac{B}{D(e_{\varepsilon})/e_{\varepsilon}} \to 0
\end{align*}

which proves the claim.
\end{proof}

Here is the proof of the second part of proposition
\ref{prop:multi-period-task-continuous}.

\begin{proof}
Since $\tau_0$ is continuous and decreasing on $(B_L(E_0), B_H(E_0))$, and since $\tau_0$ can be $0$ and $1$, we know that for every $\tau \in (0, 1)$ there is some $B_0 \in (B_L(E_0), B_H(E_0))$ such that $\tau_0(E_0, B_0) = \tau$. Notice that at time $\tau_0$ we have that $G(1 - \tau_0, 0, \frac{E_{1 - \tau_0}}{\tau_0}) = B_0$. Therefore $ \tilde{D}(\frac{E_{1 - \tau_0}}{\tau_0}|0) \tau_0 = B_0$. As $\tau_0 \to 0$, by Lemma~\ref{lemma:little-left-to-do-to} we must therefore have that $E_{1 - \tau_0} \to 0$. This means that almost all the work gets done before time $1 - \tau_0$ for which the least disutility is $D(E_0 - \varepsilon) > D(E_0) - \delta$ for sufficiently small $\varepsilon$ (i.e. $\tau_0$ sufficiently close to $1$). Therefore the disutility is at least $D(E_0) - \delta$ for arbitrary $\delta$. Hence the result holds.
\end{proof}

\subsection{Proofs for Section 5}

Proof of proposition \ref{prop:minimize-productivity-effort}.

\begin{proof}
  The agent solves the following maximization problem:
  \begin{equation*}
    \max_{\vect{e}} B - \sum_{t = 1}^{T} D(e_t) \text{, s.t. } \sum_{t = 1}^T p_t \cdot e_t = E
  \end{equation*}
  where $B$ is a fixed benefit for completing $E$ total work, and $p_t$ is the productivity in period $t$, that is the amount of effective work done for each unit of effort exerted. This means that optimal effort is determined by the following first order conditions:
  \begin{equation*}
    \frac{D'(e_t^{*})}{p_t} \geq \lambda
  \end{equation*}
  In period 1, after having worked for a time $s$, the optimal perceived effort levels are instead given by the following perceived first order conditions:
  \begin{align*}
    \frac{\tilde{D}'(\tilde{e}_t^{*}(s)|s)}{p_t} \geq \lambda(s) & \iff \frac{(1 - \alpha) \cdot D'(\tilde{e}_t^{*}(s))}{p_t} + \frac{\alpha \cdot D'(s)}{p_t} \geq \lambda (s) \\
    & \iff \frac{D'(\tilde{e}_t^{*}(s))}{p_t} \geq \frac{\lambda(s)}{1 - \alpha} - \frac{\alpha}{1 - \alpha} \frac{D'(s)}{p_t} \\
    & \iff \frac{D'(\tilde{e}_{t|1}^{*})}{p_t} \geq \frac{\lambda}{1 - \alpha} - \frac{\alpha}{1 - \alpha} \frac{D'(\tilde{e}_1^{*})}{p_t} = \lambda_t
  \end{align*}
  
  where $\tilde{e}_{t|1}^{*}$ is the amount of work the person plans to do on day $t$ at the end of day 1, and $\tilde{e}_{1}^{*}$ is the amount of work done on day 1. Note that when the productivities are strictly increasing, then so are the $\lambda_i$, and when the productivities are strictly decreasing, then so are the $\lambda_{i}$. When the productivities are strictly increasing, then $\lambda_{1} < \lambda^{*}$: suppose not, so that $\lambda_{1} \geq \lambda^{*}$. Then since the $\lambda_{i} > \lambda_{1} \geq \lambda^{*}$ for all $i > 1$, this implies that whenever the unbiased agent exerts strictly positive effort, then the biased agent plans strictly more effort on all those days, and hence strictly more total effort -- which violates the output constraint. So as long as the person exerts effort on more than a single day (the final day), this cannot be -- hence $\lambda_{1} < \lambda^{*}$, and the person works less than optimal on the first day (strictly less if $e_{1}^{*} > 0$ and $T \geq 2$).
  
  Similarly, when productivities are strictly decreasing, we have that $\lambda_{1} > \lambda^{*}$. If not, then we have that $\lambda^{*} \geq \lambda_{1} > \lambda_{2} > ... > \lambda_{T}$. Thus the person plans to do strictly less work on every future day on which it is optimal to exert strictly positive effort -- but this means that they plan on producing strictly less output than is required, which is not possible. So as long as it is optimal to exert some strictly positive amount of work on some future day, $\lambda_{1} > \lambda^{*}$ and thus the person works more on the first day (strictly more if it is optimal to work on at least 2 days).
  
  Let us now prove the following two separate statements for increasing and decreasing productivity respectively.
  
  \textbf{Case 1: Increasing productivity} Suppose that productivity is strictly increasing, so that $p_1 < p_2 < ... < p_T$. Then $\lambda(s)$ is strictly increasing in $D'(s)$ and hence in $s$.

  We will show the following in turn:

  \begin{enumerate}
    \item The agent works strictly less on the first day than they should
    \item At the start of every day after the first, the biased agent has completed strictly less work in total than the unbiased agent
    \item On every day that is not the first or the last, the agent ends up working strictly less than they expected to work on this day at the end of the day before
  \end{enumerate}

  \textbf{Step 1:} We proved this above.
  
  \textbf{Step 2:} Let $E_{t} = \sum_{i = 1}^{t} p_{i} \cdot e_{i}^{*}$ be the total work completed at the end of day $t$ by the unbiased agent and $\tilde{E}_{t} = \sum_{i = 1}^{t} p_{i} \cdot \tilde{e}_{i}^{*}$ be the total work completed at the end of day $t$ by the biased agent.

  Then we want to show that $\tilde{E}_{t} < E_{t}$ for all $t$ from $1$ to $T - 1$ -- since of course at the end of the last day, the person has completed the same amount of work. We know from step 1 that $\tilde{E}_{1} = \tilde{e}_{1}^{*} < e_{1}^{*} = E_{1}$, so the result holds for $t = 1$. We will prove it by induction. Suppose that the result holds for all $\tau < t$. If $t = T$, then we are done. If $t < T$, then we will prove that the result also holds for $\tau = t$ and hence for all $\tau \leq t$. Thus by induction it holds for all $\tau < T$, proving our claim.

  Why does the result hold for $\tau = t$? Suppose by contradiction that it does not hold. Then $\tilde{E}_{t - 1} < E_{t - 1}$ and $\tilde{E}_{t} \geq E_{t}$, hence $\tilde{e}_{t}^{*} > e_{t}^{*}$. Consider how much the biased agent would work if suddenly on day $t$ they would become unbiased: that is, they can not change the fact that they worked suboptimally in the past, but moving forward they will work optimally. Let us denote this person's variables with a hat rather than a tilde, i.e. $\hat{e}_{t}^{*}$ is the amount of effort they will exert on day $t$. Since day $t$ is the first day of the rest of their life, by step 1 we know that they would work more now than they did as a biased agent: $\hat{e}_{t}^{*} > \tilde{e}_{t}^{*} \geq e_{t}^{*}$. But this means that they started day $t$ having to complete strictly more work on the remaining days than an agent who worked optimally from the start, yet on day $t + 1$ they have strictly less work remaining than such an agent, despite both choosing optimally and having convex disutility of effort. This is a contradiction.

  Hence step 2 holds.

  \textbf{Step 3:} The FOCs for planned work on day 1 are given by:
  \begin{equation*}
    \frac{D'(\tilde{e}_{t|1}^{*})}{p_t} \geq \frac{\lambda}{1 - \alpha} - \frac{\alpha}{1 - \alpha} \frac{D'(\tilde{e}_1^{*})}{p_t} = \lambda_t
  \end{equation*}
  whereas the FOCs for planned work on day 2 are given by:
  \begin{equation*}
    \frac{D'(\tilde{e}_{t|2}^{*})}{p_t} \geq \frac{\mu}{1 - \alpha} - \frac{\alpha}{1 - \alpha} \frac{D'(\tilde{e}_2^{*})}{p_t} = \mu_t
  \end{equation*}
  Our claim is that $\lambda_{2} > \mu_{2}$, since then $\tilde{e}_{2|1}^{*} \geq \tilde{e}_{2}^{*}$, with strict inequality if $\tilde{e}_{2|1}^{*} > 0$. But given that actual effort on day $2$ is certainly higher than actual effort on day 1 (note that as long as the person has worked less than they worked on day 1, they are planning more work on day $2$ than they planned), final effort on day 2 is higher. Note that both at the end of day 1 and at the beginning of day 2, the person plans to produce identical total output over the remaining $T - 1$ days. Now suppose that we have $\mu_{2} \geq \lambda_{2}$. We have:
  \begin{align}
    \lambda_t & = \frac{\lambda}{1 - \alpha} - \frac{\alpha}{1 - \alpha} \frac{D'(\tilde{e}_1^{*})}{p_t} \nonumber \\
    & = \lambda_2 - \frac{\alpha \cdot D'(\tilde{e}_1^{*})}{1 - \alpha} (\frac{1}{p_t} - \frac{1}{p_2}) \nonumber \\
    & \leq \mu_2 -  \frac{\alpha \cdot D'(\tilde{e}_1^{*})}{1 - \alpha} (\frac{1}{p_t} - \frac{1}{p_2}) \label{eq:productivity-overestimate}\\
    & < \mu_2 -  \frac{\alpha \cdot D'(\tilde{e}_2^{*})}{1 - \alpha} (\frac{1}{p_t} - \frac{1}{p_2}) \nonumber \\
    & = \mu_t \nonumber
  \end{align}
  where the first inequality holds because we assumed that $\mu_{2} \geq \lambda_{2}$, and the second holds because $1/p_t - 1/p_2 < 0$ and $\tilde{e}_{2}^{*} \geq \tilde{e}_1^{*}$. But if $\mu_t > \lambda_t$ for all $t > 2$ and $\mu_2 \geq \lambda_2$, then the person plans to do more work on every single day and hence produce more output in total on day 2 than they planned at the end of day 1. As long as they plan to do some positive amount of work on some day (which they do on the last day), they plan strictly more work, hence this is a contradiction, unless the second day is the last day. Thus $\mu_2 < \lambda_2$ and the person works strictly less on the second day than they planned if they planned on doing a non-zero amount of work.
  
  \textbf{Case 2} Let us now assume that the productivity strictly \emph{decreases} over time. We will show the following in turn:
  \begin{enumerate}
    \item The agent works strictly more on the first day than they should
    \item At the start of every day after the first, the biased agent has completed strictly more work in total than the unbiased agent
    \item On every day that is not the first or the last, the agent ends up working strictly less than they expected to work on this day at the end of the day before
  \end{enumerate}
  \textbf{Step 1} We showed this before.

  \textbf{Step 2} This is essentially identical to the argument in the previous step 2.

  \textbf{Step 3} The argument mirrors step 3 from the increasing productivity case, but somewhat surprising maybe, the person does still overestimate the amount of work they will do. Thus we again want to show that $\lambda_{2} > \mu_{2}$. In fact, the derivation in the equations in \ref{eq:productivity-overestimate} still hold: the first inequality holds if we assume (by contradiction) that $\mu_{2} \geq \lambda_{2}$; the second inequality holds because we now have that $1/p_t - 1/p_2 > 0$, but we also have that $\tilde{e}_{2}^{*} < \tilde{e}_{1}^{*}$ (as long as there was any effort on day 1) -- after all, productivity is higher on day 1, so the person works more on day 1 than on day 2. Thus the result follows.
\end{proof}

\hypertarget{bibliography}{%
\section*{References}\label{bibliography}}
\addcontentsline{toc}{section}{References}

\hypertarget{refs}{}
\leavevmode\hypertarget{ref-aclandLevy2015naivete}{}%
Acland, Dan, and Matthew R. Levy. 2015. ``Naiveté, Projection Bias, and
Habit Formation in Gym Attendance.'' \emph{Management Science} 61 (1):
146--160.

\leavevmode\hypertarget{ref-ahnIijimaSarver2020recursiveTemptation}{}%
Ahn, David S, Ryota Iijima, and Todd Sarver. 2020. ``Naivete About
Temptation and Self-Control: Foundations for Recursive Naive
Quasi-Hyperbolic Discounting.'' \emph{Journal of Economic Theory} 189.
Elsevier: 105087.

\leavevmode\hypertarget{ref-arielyLoewenstein2006heat}{}%
Ariely, Dan, and George Loewenstein. 2006. ``The Heat of the Moment: The
Effect of Sexual Arousal on Sexual Decision Making.'' \emph{Journal of
Behavioral Decision Making} 19 (2): 87--98.

\leavevmode\hypertarget{ref-augenblickRabin2019unpleasant}{}%
Augenblick, Ned, and Matthew Rabin. 2019. ``An Experiment on Time
Preference and Misprediction in Unpleasant Tasks.'' \emph{Review of
Economic Studies} 86 (3). Oxford University Press: 941--975.

\leavevmode\hypertarget{ref-badger2007heroin}{}%
Badger, Gary J., Warren K. Bickel, Louis A. Giordano, Eric A. Jacobs,
George Loewenstein, and Lisa Marsch. 2007. ``Altered States: The Impact
of Immediate Craving on the Valuation of Current and Future Opioids.''
\emph{Journal of Health Economics} 26 (5): 865--876.

\leavevmode\hypertarget{ref-bray2016multitasking}{}%
Bray, Robert L, Decio Coviello, Andrea Ichino, and Nicola Persico. 2016.
``Multitasking, Multiarmed Bandits, and the Italian Judiciary.''
\emph{Manufacturing \& Service Operations Management} 18 (4). Informs:
545--558.

\leavevmode\hypertarget{ref-buchheimKolaska2017weather}{}%
Buchheim, Lukas, and Thomas Kolaska. 2017. ``Weather and the Psychology
of Purchasing Outdoor Movie Tickets.'' \emph{Management Science} 63
(11). Informs: 3718--3738.

\leavevmode\hypertarget{ref-buehler1994exploring}{}%
Buehler, Roger, Dale Griffin, and Michael Ross. 1994. ``Exploring the"
Planning Fallacy": Why People Underestimate Their Task Completion
Times.'' \emph{Journal of Personality and Social Psychology} 67 (3).
American Psychological Association: 366.

\leavevmode\hypertarget{ref-bushongGagnonBartsch2020experimentInterpersonalProjectionBias}{}%
Bushong, Benjamin, and Tristan Gagnon-Bartsch. 2020. ``An Experiment on
Interpersonal Projection Bias.''

\leavevmode\hypertarget{ref-bussePopePopeSilvaRisso2015convertibles}{}%
Busse, Meghan R, Devin G Pope, Jaren C Pope, and Jorge Silva-Risso.
2015. ``The Psychological Effect of Weather on Car Purchases.''
\emph{The Quarterly Journal of Economics} 130 (1). Oxford University
Press: 371--414.

\leavevmode\hypertarget{ref-chaloupkaLevyFrank2020smokingCessation}{}%
Chaloupka, IV, Frank J, Matthew R Levy, and Justin S White. 2019.
\emph{Estimating Biases in Smoking Cessation: Evidence from a Field
Experiment}. Working Paper 26522. Working Paper Series. National Bureau
of Economic Research.
doi:\href{https://doi.org/10.3386/w26522}{10.3386/w26522}.

\leavevmode\hypertarget{ref-conlinOdonoghueVogelsang2007catalogOrders}{}%
Conlin, Michael, Ted O'Donoghue, and Timothy J. Vogelsang. 2007.
``Projection Bias in Catalog Orders.'' \emph{American Economic Review}
97 (4): 1217--1249.

\leavevmode\hypertarget{ref-covielloIchinoPersico2014timeAllocationAndTaskJuggling}{}%
Coviello, Decio, Andrea Ichino, and Nicola Persico. 2014. ``Time
Allocation and Task Juggling.'' \emph{American Economic Review} 104 (2):
609--623.
doi:\href{https://doi.org/10.1257/aer.104.2.609}{10.1257/aer.104.2.609}.

\leavevmode\hypertarget{ref-covielloIchinoPersico2015inefficiencyOfWorkerTimeUse}{}%
Coviello, Decio, Andrea Ichino, and Nicola Persico. 2015. ``The
Inefficiency of Worker Time Use.'' \emph{Journal of the European
Economic Association} 13 (5): 906--947.
doi:\href{https://doi.org/10.1111/jeea.12129}{10.1111/jeea.12129}.

\leavevmode\hypertarget{ref-ericson2019intertemporal}{}%
Ericson, Keith Marzilli, and David Laibson. 2019. ``Intertemporal
Choice.'' In \emph{Handbook of Behavioral Economics: Applications and
Foundations 1}, 2:1--67. Elsevier.

\leavevmode\hypertarget{ref-fedyk2018asymmetric}{}%
Fedyk, Anastassia. 2018. ``Asymmetric Naivete: Beliefs About
Self-Control.'' \emph{Available at SSRN 2727499}.

\leavevmode\hypertarget{ref-gagnon2019learning}{}%
Gagnon-Bartsch, Tristan, and Benjamin Bushong. 2019. \emph{Learning with
Misattribution of Reference Dependence}. Technical report, Michigan
State University.

\leavevmode\hypertarget{ref-gilbert1998immuneNeglect}{}%
Gilbert, Daniel T., Elizabeth C. Pinel, Timothy D. Wilson, Stephen J.
Blumberg, and Thalia P. Wheatley. 1998. ``Immune Neglect: A Source of
Durability Bias in Affective Forecasting.'' \emph{Journal of Personality
and Social Psychology} 75 (3): 617.

\leavevmode\hypertarget{ref-gul2001temptation}{}%
Gul, Faruk, and Wolfgang Pesendorfer. 2001. ``Temptation and
Self-Control.'' \emph{Econometrica} 69 (6). Wiley Online Library:
1403--1435.

\leavevmode\hypertarget{ref-haggag2019attribution}{}%
Haggag, Kareem, Devin G Pope, Kinsey B Bryant-Lees, and Maarten W Bos.
2019. ``Attribution Bias in Consumer Choice.'' \emph{The Review of
Economic Studies} 86 (5). Oxford University Press: 2136--2183.

\leavevmode\hypertarget{ref-harrisLaibson2013instantaneousGratification}{}%
Harris, Christopher, and David Laibson. 2013. ``Instantaneous
Gratification.'' \emph{The Quarterly Journal of Economics} 128 (1). MIT
Press: 205--248.

\leavevmode\hypertarget{ref-herrnsteinPrelec1991melioration}{}%
Herrnstein, Richard J, and Dražen Prelec. 1991. ``Melioration: A Theory
of Distributed Choice.'' \emph{The Journal of Economic Perspectives}.
JSTOR, 137--156.

\leavevmode\hypertarget{ref-hsiaw2013goalSettingAndSelfControl}{}%
Hsiaw, Alice. 2013. ``Goal-Setting and Self-Control.'' \emph{Journal of
Economic Theory} 148 (2). Elsevier: 601--626.

\leavevmode\hypertarget{ref-huangNguyen-Huu2018timeConsistentStopping}{}%
Huang, Yu-Jui, and Adrien Nguyen-Huu. 2018. ``Time-Consistent Stopping
Under Decreasing Impatience.'' \emph{Finance and Stochastics} 22 (1).
Springer: 69--95.

\leavevmode\hypertarget{ref-kaufmann2020wpProjectionBias}{}%
Kaufmann, Marc. 2020. ``Projection Bias in Effort Choices.'' Working
Paper.
\url{https://trichotomy.xyz/publication/projection-bias-in-effort-choices/projection-bias-in-effort-choices.pdf}.

\leavevmode\hypertarget{ref-laibson1997golden}{}%
Laibson, David. 1997. ``Golden Eggs and Hyperbolic Discounting.''
\emph{The Quarterly Journal of Economics} 112 (2). MIT Press: 443--478.

\leavevmode\hypertarget{ref-levy2009empirical}{}%
Levy, Matthew. 2009. ``An Empirical Analysis of Biases in Cigarette
Addiction.'' Working Paper.

\leavevmode\hypertarget{ref-leYaouanqSchwardmann2019learning}{}%
Le Yaouanq, Yves, and Peter Schwardmann. 2019. ``Learning About One's
Self.'' CEPR Discussion Paper No. DP13510.

\leavevmode\hypertarget{ref-loewensteinAdler1995bias}{}%
Loewenstein, George, and Daniel Adler. 1995. ``A Bias in the Prediction
of Tastes.'' \emph{The Economic Journal}. JSTOR, 929--937.

\leavevmode\hypertarget{ref-loewensteinNaginPaternoster1997sexualArousal}{}%
Loewenstein, George, Daniel Nagin, and Raymond Paternoster. 1997. ``The
Effect of Sexual Arousal on Expectations of Sexual Forcefulness.''
\emph{Journal of Research in Crime and Delinquency} 34 (4). Sage
Publications: 443--473.

\leavevmode\hypertarget{ref-loewenstein2003projection}{}%
Loewenstein, George, Ted O'Donoghue, and Matthew Rabin. 2003.
``Projection Bias in Predicting Future Utility.'' \emph{The Quarterly
Journal of Economics} 118 (4): 1209--1248.

\leavevmode\hypertarget{ref-michelStenzel2020coolingOffPolicies}{}%
Michel, Christian, and André Stenzel. 2020. ``Model-Based Evaluation of
Cooling-Off Policies.''

\leavevmode\hypertarget{ref-nordgren2008instability}{}%
Nordgren, Loran F, Joop van der Pligt, and Frenk van Harreveld. 2008.
``The Instability of Health Cognitions: Visceral States Influence
Self-Efficacy and Related Health Beliefs.'' \emph{Health Psychology} 27
(6). American Psychological Association: 722.

\leavevmode\hypertarget{ref-odonoghueRabin1999doingNowOrLater}{}%
O'Donoghue, Ted, and Matthew Rabin. 1999. ``Doing It Now or Later.''
\emph{The American Economic Review} 89 (1): 103--124.

\leavevmode\hypertarget{ref-odonoghueRabin2008longTermProcrastination}{}%
O'Donoghue, Ted, and Matthew Rabin. 2008. ``Procrastination on Long-Term
Projects.'' \emph{Journal of Economic Behavior \& Organization} 66 (2):
161--175.

\leavevmode\hypertarget{ref-quahStrulovici2013discounting}{}%
Quah, John K-H, and Bruno Strulovici. 2013. ``Discounting, Values, and
Decisions.'' \emph{Journal of Political Economy} 121 (5). University of
Chicago Press Chicago, IL: 896--939.

\leavevmode\hypertarget{ref-readVanLeeuwen1998predicting}{}%
Read, Daniel, and Barbara Van Leeuwen. 1998. ``Predicting Hunger: The
Effects of Appetite and Delay on Choice.'' \emph{Organizational Behavior
and Human Decision Processes} 76 (2): 189--205.

\leavevmode\hypertarget{ref-vanBovenLoewenstein2003social}{}%
Van Boven, Leaf, and George Loewenstein. 2003. ``Social Projection of
Transient Drive States.'' \emph{Personality and Social Psychology
Bulletin} 29 (9). Sage Publications: 1159--1168.

\end{document}